\documentclass[runningheads]{llncs}
\usepackage[T1]{fontenc}
\usepackage{graphicx}
\usepackage{amsmath}
\usepackage{amsfonts}
\usepackage{amssymb}
\usepackage{mathtools}
\usepackage{dsfont}
\usepackage{hyperref}
\usepackage[nameinlink,noabbrev]{cleveref}
\usepackage{complexity}
\usepackage{tabularx}
\usepackage{booktabs}
\usepackage[disable]{todonotes}

\newtheorem{ruleenv}{Rule}
\Crefname{ruleenv}{Rule}{Rules}

\usepackage{thmtools}
\usepackage{thm-restate}

\usepackage{environ} 

\newif\iflong
\newif\ifshort

\ifshort
\longfalse
\else
\longtrue
\fi

\newif\ifhideproofs

\ifshort
\hideproofstrue 
\fi

\makeatletter
\let\proof@orig\proof
\let\endproof@orig\endproof

\newenvironment{keepproof}[1][]%
  {\proof@orig[#1]}{\endproof@orig}

\ifhideproofs
  \RenewEnviron{proof}[1][]{%
  }
\else
  \RenewEnviron{proof}[1][]{\proof@orig[#1]\BODY\endproof@orig}
\fi
\makeatother

\DeclareMathOperator{\anc}{anc}
\DeclareMathOperator{\desc}{desc}

\usepackage{xparse}

\makeatletter
\newcommand\saveTheoremNumber[1]{%
  \expandafter\xdef\csname saved@thm@#1\endcsname{\number\value{theorem}}%
  \@ifundefined{theHtheorem}{}{%
    \expandafter\xdef\csname saved@Hthm@#1\endcsname{\theHtheorem}%
  }%
}

\NewDocumentEnvironment{theoremreused}{m o}{%
  \begingroup
  \@ifundefined{saved@thm@#1}{%
    \PackageError{theoremreused}{No saved theorem number for `#1'. Did you call \string\saveTheoremNumber{#1}?}{}%
  }{}%
  \setcounter{theorem}{\numexpr\csname saved@thm@#1\endcsname-1\relax}%
  \@ifundefined{theHtheorem}{}{%
    \edef\theHtheorem{\csname saved@Hthm@#1\endcsname R}%
  }%
  \IfNoValueTF{#2}{\begin{theorem}}{\begin{theorem}[#2]}%
}{%
  \end{theorem}%
  \endgroup
}
\makeatother

\begin{document}

\title{The Parameterized Complexity of Geometric 1-Planarity\thanks{The author acknowledges support from the Vienna Science and Technology Fund (WWTF)
[10.47379/ICT22029] and the Austrian Science Fund (FWF) [10.55776/Y1329], and thanks Robert
Ganian for helpful discussions.}}

\author{Alexander Firbas\orcidID{0009-0007-2049-2144}}
\authorrunning{A. Firbas}
\institute{TU Wien, Austria\\
\email{afirbas@ac.tuwien.ac.at}}

\maketitle

\begin{abstract}
A graph is \emph{geometric 1-planar} if it admits a straight-line drawing where each edge is crossed at most once.  
We provide the first systematic study of the parameterized complexity of recognizing geometric 1-planar graphs.
By substantially extending a technique of Bannister, Cabello, and Eppstein, combined with Thomassen’s characterization of 1-planar embeddings that can be straightened, we show that the problem is fixed-parameter tractable when parameterized by treedepth.  
Furthermore, we obtain a kernel for \textsc{Geometric 1-Planarity} parameterized by the feedback edge number $\ell$.
As a by-product, we improve the best known kernel size of $\mathcal{O}((3\ell)!)$ for \textsc{1-Planarity}~\cite{1planar_parameterized} and \textsc{$k$-Planarity}~\cite{k_planarity_gd} under the same parameterization to $\mathcal{O}(\ell \cdot 8^{\ell})$.  
Our approach naturally extends to \textsc{Geometric $k$-Planarity}, yielding a kernelization under the same parameterization, albeit with a larger kernel.
Complementing these results, we provide matching lower bounds: \textsc{Geometric 1-Planarity} remains \NP-complete even for graphs of bounded pathwidth, bounded feedback vertex number, and bounded bandwidth.

\keywords{geometric 1-planarity \and parameterized complexity \and graph drawing}
\end{abstract}

\section{Introduction}

A graph is \emph{1-planar} if it admits a drawing where each edge is crossed at most once, and \emph{geometric 1-planar} if it admits such a drawing with straight-line edges. Recognizing 1-planar graphs is \NP-complete under various restrictions; see the survey by Kobourov, Liotta, and Montecchiani~\cite{survey}. Despite this hardness, the parameterized complexity of \textsc{1-Planarity} is comparatively well understood. In their influential work, Bannister, Cabello, and Eppstein analyzed structural parameterizations, giving an essentially tight classification of which parameters yield fixed-parameter tractability (e.g., treedepth, feedback edge number) and which retain \NP-completeness (e.g., bandwidth, pathwidth)~\cite{1planar_parameterized}. This has been recently extended to \textsc{$k$-Planarity}~\cite{k_planarity_gd}, and there is complementary work that studies orthogonal parameterizations, such as completion from partially predrawn instances~\cite{extending1planar} and the total number of crossings~\cite{PartiallyPredrawnCrossingNumber}.

The topological and geometric settings exhibit different combinatorial and algorithmic behavior.
Geometric 1-planar graphs on \(n\) vertices have at most \(4n-9\) edges (tight for infinitely many \(n\ge 8\))~\cite{density_geometric_1_planar}. By comparison, 1-planar graphs admit up to \(4n-8\) edges (tight for \(n\ge 12\))~\cite{survey}. Hence, not every 1-planar graph is geometric 1-planar.

From a complexity perspective, there is strong evidence that the geometric variant is strictly harder: for every fixed \(k \ge 1\), \textsc{Geometric \(k\)-Planarity} is \NP-hard~\cite{geometric_1_planarity_np_hard,geometric_k_planar_complexity}, and for \(k=867\) it is \(\exists\mathbb{R}\)-complete~\cite{geometric_k_planar_complexity}, implying that, for this fixed \(k\), the problem is not even in \NP{} unless \(\exists\mathbb{R}=\NP\). In contrast, the topological counterpart \textsc{\(k\)-Planarity} is trivially in \NP{} for every fixed \(k\).

\begin{figure}[t]
    \centering
    \includegraphics[page=7]{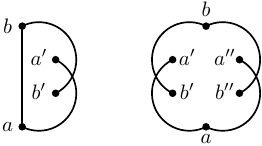}
\caption{Thomassen’s $B$ (left) and $W$ (right) configurations~\cite{straightening_characterization}; see \Cref{definition:bw_configs} for their formal definition.}
    \label{figure:bw_configs}
\end{figure}

Central to understanding geometric 1-planarity is Thomassen’s straightening characterization: a 1-planar embedding can be straightened if and only if it contains no \(B\)- or \(W\)-configurations (see \Cref{figure:bw_configs}). This result was first conjectured by Eggleton~\cite{Eggleton1986}, proved by Thomassen in 1988~\cite{straightening_characterization}, and later rediscovered by Hong et al.~\cite{rediscovered}, who also showed the \(B/W\) configurations can be detected in linear time in a given embedding, yielding an \NP-certificate for \textsc{Geometric 1-Planarity}. No straightening characterization is known for 2-planar embeddings and moreover, any such characterization would have to be infinite~\cite{straightening_characterization}. For 3-planar embeddings, in a certain natural sense, no such characterization exists~\cite{nagamochi2013straight}. If one relaxes topological equivalence and only preserves the set of crossing pairs in a 1-planar embedding, a different characterization is available~\cite{reembedding_crossing_pairs}.

Two algorithmic applications enabled by Thomassen’s characterization are worth noting. First, every triconnected 1-planar graph admits a drawing that is straight-line except for one bent edge on the outer face~\cite{grid_drawings}. Second, IC-planar graphs (1-planar graphs admitting an embedding in which only independent edges cross) can be drawn straight-line in linear time~\cite{ic_planar}.

Beyond recognition and concrete drawing algorithms, there is also a quantitative perspective on straight-line drawings. The \emph{rectilinear local crossing number} \(\overline{\textup{lcr}}(\cdot)\) measures how many crossings per edge are required to draw a graph with straight-line edges~\cite{schaefer2012graph}. It is known exactly for all complete graphs~\cite{rectilinear_crossing_number_complete} and for most complete bipartite graphs~\cite{rectilinear_crossing_number_bipartite}.

Finally, in recent parallel and independent work, Cabello et al.~\cite{crossing_types}
study \textsc{1-Planarity} (and its geometric variant) restricted to a subset of allowed crossing types,
where a crossing type is defined by the subgraph induced by the endpoints of the two crossing
edges. They characterize for which subsets the problem admits a fixed-parameter
algorithm parameterized by treewidth and for which subsets it remains \NP-hard even for graphs of
constant pathwidth. In particular, for the unrestricted problem \textsc{Geometric 1-Planarity}
considered here (i.e., allowing all crossing types), their results imply hardness.
This naturally raises the question of how the problem behaves for structural parameters that are
more restrictive than treewidth, such
as treedepth.

\subsection{Contributions}

As our first result, we obtain:
\begin{restatable}{theorem}{fpttreedepth}\label{theorem:fpt_treedepth}
Let $G$ be a graph on $n$ vertices with treedepth at most $d$.
Then \textsc{Geometric 1-Planarity} can be decided in time
\[
\mathcal{O}\!\bigl(2^{\,2^{\,2^{\,2^{\mathcal{O}(d)}}}}\cdot n^{\mathcal{O}(1)}\bigr).
\]
\end{restatable}
\saveTheoremNumber{theorem:fpt_treedepth}

When deriving a treedepth-\FPT{} algorithm, one essentially aims, given some set of vertices whose removal disconnects the graph (sometimes called a \emph{modulator}) to bound the number of components that connect to the remaining graph only via this set by a function depending only on the treedepth.

\looseness=-1
For topological 1-planarity, the algorithm of Bannister, Cabello, and Eppstein~\cite{1planar_parameterized} essentially considers two cases: First, if the modulator has at least three vertices, each component induces a ``claw'' graph \(K_{1,3}\) connected to the modulator. Since bounded treedepth implies paths are of bounded length, each claw has bounded size. Hence, too many such components will create too many crossings, contradicting 1-planarity. Second, if the modulator has size two, again from the fact that the path length is bounded, if there are more than a certain number of such components, one can show that the modulator vertices need to be drawn in a shared region, i.e., one can draw an uncrossed line between them. Thus, the instance decomposes into independent instances: if there is a way to draw each component with the two modulator vertices on the same region, one can ``glue'' them back together. The case of modulator size one does not need to be considered, since in the topological case, via a simple ``gluing'' argument, one can treat each biconnected component (i.e., each \emph{block}) separately.

\looseness=-1
The correctness of gluing hinges on one important assumption: for a topological 1-planar drawing, the choice of outer region is immaterial. That is, two 1-planar drawings can always be glued at a shared vertex without creating new crossings by selecting the outer regions of the subdrawings accordingly. For geometric 1-planarity, this freedom of choice of the outer region vanishes. Indeed, the straightening characterization of Thomassen~\cite{straightening_characterization} is sensitive to the choice of outer region.

As we cannot assume biconnectivity, our algorithm works in two stages.

(i) \emph{Inside blocks}, we process the graph along a treedepth decomposition. The \emph{3-modulator case} works as in the topological setting. For \emph{2-modulators}, we process bottom–up at each node, grouping children by their common two-vertex attachment. We retain only a small baseline number of such children so that any valid solution must place the two attachment vertices on a shared region. Among the remaining children we decide by brute force which ones admit a drawing with those two vertices on the \emph{outer} region; these pieces are “glueable’’—by a nontrivial application of Thomassen’s characterization they can always be reinserted without creating new crossings—and we discard them. If there are too many of the remaining (non-outer) pieces, some edge would be forced to receive more than one crossing, contradicting 1-planarity; otherwise only a bounded number remain. This bottom–up filtering bounds, as a function of the treedepth, the number of vertices inside each block.

\looseness=-1
(ii) \emph{Across blocks}, we work on the block–cut tree, again bottom–up. At a cut vertex we examine each child subgraph (the union of blocks below that child) and, by brute force, decide which ones admit a drawing with the cut vertex on the \emph{outer} region; these are glueable at the cut vertex and can be safely deleted. If too many non-outer children remain, a 1-planar drawing would be impossible, so we reject.
This processing allows us to bound the maximum degree of the block-cut tree.
Since paths are of bounded length, the height of the block-cut tree is bounded as well.
In total, this allows us to bound the number of blocks (which in turn have only a bounded number of vertices) solely in terms of the treedepth.

\medskip
For our next set of results, we consider a parameter incomparable with treedepth, the feedback edge number, also referred to as the cyclomatic number.
First, we consider topological 1-planarity.
The technique of Bannister, Cabello, and Eppstein for their feedback-edge-number kernel yields a kernel of size $\mathcal{O}((3\ell)!)$, where $\ell$ is the feedback edge number.
Essentially, they decompose the input graph into $\mathcal{O}(\ell)$ degree-2 paths and show that, in any hypothetical solution, once a certain threshold of crossings is exceeded, one finds a \emph{local} configuration of consecutive crossings that can be redrawn to reduce the number of crossings.
Thus one obtains a kernel by shortening the degree-2 paths so they do not exceed the threshold.
They further show that this kernel size is optimal under this strategy \cite[Lemma 10]{1planar_parameterized}.

We break this barrier using a \emph{global}, rather than \emph{local}, redrawing argument.
We also decompose the input graph into degree-2 paths. Then, we order them by length, and observe a qualitative shift once there is a sufficiently large gap: the \emph{long} paths are long enough to interact arbitrarily with the \emph{short} paths, and any two long paths can be redrawn to cross at most once by applying \emph{Reidemeister moves}, a foundational tool in knot theory \cite{reidemeister}.
We then obtain our kernel by shortening the long paths to the fewest edges that still qualify them as “long.”
Using a recursive bound and an observation that allows us to assume the shortest path has bounded length, we obtain an $\mathcal{O}(\ell \cdot 8^{\ell})$-edge kernel for \textsc{1-Planarity}.

Moreover, at the cost of a higher base in the exponent, applying Thomassen’s characterization, we extend this kernel to the geometric case.

\begin{restatable}{theorem}{fenfpt}\label{theorem:fen_fpt}
\textsc{1-Planarity}, parameterized by the feedback edge number~$\ell$,
admits a kernel with $\mathcal{O}(\ell \cdot 8^\ell)$ edges.
\textsc{Geometric 1-Planarity}, under the same parameterization,
admits a kernel with $\mathcal{O}(\ell \cdot 27^\ell)$ edges.
\end{restatable}
\saveTheoremNumber{theorem:fen_fpt}

As an immediate corollary (\Cref{corollary:fen_k_planarity}), we obtain an improved kernel for \textsc{$k$-Planarity} under the same parametrization with $\mathcal{O}(\ell \cdot 8^{\ell})$ edges, thereby improving upon the previous best-known kernel of size $\mathcal{O}((3\ell)!)$ due to Gima, Kobayashi, and Okada~\cite{k_planarity_gd}.

Our technique also extends to \textsc{Geometric $k$-Planarity}, at the expense of a larger kernel size.
Due to the lack of a straightening characterization for \textsc{$k$-Planarity} with $k>1$, we provide a direct redrawing argument: we triangulate the graph with respect to the short paths and redraw the long paths inside each triangle with few crossings.
Thus we obtain:
\begin{restatable}{theorem}{geomkplanarfpt}\label{theorem:geom_k_planar_fen_fpt}
\textsc{Geometric $k$-Planarity}, parameterized by the feedback edge number~$\ell$,
admits a kernel with $\mathcal{O}\big(2^{\mathcal{O}(3^{\ell}\log \ell)}\big)$ edges.
\end{restatable}
\saveTheoremNumber{theorem:geom_k_planar_fen_fpt}

\medskip
Finally, we show that our results are essentially tight within the usual parameter hierarchy:
\textsc{Geometric 1-Planarity} remains \NP-complete even in very restricted settings.\footnote{In parallel and independent work, Cabello et al.~\cite{crossing_types}
also show \NP-hardness for \textsc{Geometric 1-Planarity} on graphs whose pathwidth is bounded by a constant, but without specifying an explicit bound on that constant.}
We provide a novel reduction from \textsc{Bin Packing} and show:
\begin{restatable}{theorem}{npcbyfvs}\label{theorem:npc_by_fvs}
    \textsc{Geometric 1-Planarity} remains \NP-complete
    for instances of pathwidth at most~15 or feedback vertex number at most~48.
\end{restatable}
\saveTheoremNumber{theorem:npc_by_fvs}

For the parameter bandwidth, we observe that replacing every edge by a constant-size gadget causes only a quadratic blow-up in bandwidth. Combining this with a known reduction of Schaefer~\cite{geometric_1_planarity_np_hard} lets us lift the known bounded-bandwidth hardness for \textsc{1-Planarity}~\cite{1planar_parameterized}:
\begin{restatable}{theorem}{npcbybandwidth}\label{theorem:npc_by_bandwidth}
    \textsc{Geometric 1-Planarity} remains \NP-complete
    even when restricted to instances of bounded bandwidth.
\end{restatable}

\section{Preliminaries}

We assume familiarity with standard graph terminology~\cite{diestel} and the basics of parameterized complexity theory~\cite{fptbook}.
For $k \in \mathbb{N}$, $[k]$ denotes the set $\{1, \dots, k\}$.

\paragraph{Graphs and their embeddings.}
We work with finite simple graphs. For a graph $G$, we denote its vertex and edge sets by $V(G)$ and $E(G)$.  
A \emph{Jordan arc} is the image of a continuous injective map $[0,1]\to\mathbb{R}^2$.

An \emph{embedding} of $G$ is a drawing in the plane where vertices are distinct points, each edge is a Jordan arc between its endpoints, and edges intersect only at common endpoints or in proper crossings.  
A \emph{$k$-planar embedding} is an embedding in which every edge is crossed at most $k$ times.
In a \emph{geometric} embedding, each edge is drawn as a straight line. 

The regions of an embedding $\varepsilon$ are the connected components of 
$\mathbb{R}^2$ minus the union of all edge arcs; 
the unbounded one is the \emph{outer region}.
The \emph{planarization} of a 1-planar embedding $\varepsilon$ is the plane graph obtained by subdividing every crossing point into a new degree-4 \emph{dummy} vertex, whose adjacent edges we call \emph{half-edges}. We call non-dummy vertices \emph{real}.
Thus the regions of $\varepsilon$ correspond bijectively to the faces of the planarization.

\begin{definition}
    Let $G$ be a graph and let $a,b$ be distinct vertices of $G$.
    We say $G$ is $(a,b)$-shared geometric 1-planar if there is a geometric 1-planar embedding of $G$ where $a$ and $b$ are in a shared region, i.e., one can draw a Jordan arc between them without crossing any edge.
    If this shared region is unbounded, we say $G$ is $(a,b)$-outer geometric 1-planar.
    Furthermore, we say $G$ is $a$-outer geometric 1-planar if there is a geometric 1-planar embedding of $G$ where $a$ lies on the outer region.
\end{definition}

\paragraph{Crossings and their orientation.}
If two independent edges $aa',bb'\in E(G)$ cross in a 1-planar embedding $\varepsilon$,  
we call $\{aa',bb'\}$ an \emph{$(a,b)$-crossing pair}.  
Let $c$ denote the dummy vertex in the planarization.
The clockwise cyclic order of $c$'s neighbors is either $(b',a',b,a)$ or $(a',b',a,b)$.  
We call the crossing \emph{$(a,b)$-left-crossing} in the former case,
and \emph{$(a,b)$-right-crossing} otherwise.
Next, we formally define Thomassen's \(B\)- and \(W\)-configurations in our terminology.

\begin{definition}\label{definition:bw_configs}
Let $\varepsilon$ be a 1-planar embedding.
\begin{itemize}
    \item A \emph{B-configuration} consists of an  edge $ab$ and an $(a,b)$-crossing pair $aa',bb'$, such that  of the two regions delimited by the arcs (one is the outer region, the other is bounded), the endpoints $a',b'$ both lie in the bounded region. 

    \item A \emph{W-configuration} consists of two $(a,b)$-crossing pairs $aa',bb'$ and $aa'',bb''$, such that of the two regions delimited by the arcs, $a',a'',b',b''$ all lie in the bounded region.
\end{itemize}
\end{definition}
See \Cref{figure:bw_configs} for an illustration.
Thomassen~\cite{straightening_characterization} proved that a 1-planar embedding can be transformed into a geometric 1-planar embedding that is topologically equivalent (i.e., preserves the cyclic orders of edges at all vertices and crossings) if and only if it is free of $B$- and $W$-configurations.

\section{Fixed-Parameter Tractability via Treedepth}
In this section, we derive \Cref{theorem:fpt_treedepth}.
First, in \Cref{section:prelims_phase_1}, we provide the lemmas we will use to show the safety of \Cref{rule:I,rule:II}, constituting Phase~I of our algorithm.
In \Cref{section:phase_1} we define \Cref{rule:I,rule:II}, show they are safe (\Cref{lemma:td_correctness_rule_1_and_2}), and that the maximum block size is bounded after applying the rules exhaustively (\Cref{lemma:td_block_size}).
Then, in \Cref{section:prelims_phase_2}, we provide the lemmas used to show the safety of \Cref{rule:III}, constituting Phase~II of our algorithm.
Finally, in \Cref{section:phase_2} we define \Cref{rule:III}, show it is safe (\Cref{lemma:td_correctness_rule_3}), and that the maximum number of blocks is bounded after applying the rule exhaustively (\Cref{lemma:bounded_block_cut_tree}). Finally, we obtain \Cref{theorem:fpt_treedepth}.

\subsection{Preliminaries for Phase~I}
\label{section:prelims_phase_1}

Bannister, Cabello, and Eppstein proved the following result, albeit in asymptotic form.
To obtain an explicit bound for $k$ to be used in \Cref{rule:I}, we give an alternative argument.
\begin{lemma}[\textnormal{Adapted from \protect\cite[Lemma~6, Case $|S|\geq 3$]{1planar_parameterized}}]\label{lemma:bounded_claws}
    Let $G_1, \ldots, G_k$ be connected graphs with
    pairwise-disjoint vertex sets except that $a,b,c$ belong to every $G_i$, and assume $G_i-\{a,b,c\}$ is connected for each $i$.
    Let $G \coloneqq \bigcup_{i=1}^{k}G_i$.
    If $G$ has treedepth at most $d$ and $k \ge 2^{d+1}+3$, then $G$ is not 1-planar.
\end{lemma}
\begin{proof}
Towards a contradiction, let $\varepsilon$ be a 1-planar embedding of $G$.
For each $i$, let $\varepsilon_i$ be the subdrawing
induced by an edge-subset minimal graph connecting $a,b,c$ in $G_i$. Observe that this is always a  subdivided $K_{1,3}$.

Fix distinct $p,q$ and set $H:=\varepsilon_p\cup\varepsilon_q$.
Since $\operatorname{td}(G)\le d$, every simple path in $H$ has length less than $2^{d}$.
Observe that $H$ can be decomposed into a cycle and a path. Hence, $H$ consists of at most $2^{d+1}$ edges.

Every other subdrawing $\varepsilon_\ell$, $\ell \not\in \{p,q\}$, crosses $H$ at least once by Kuratowski's theorem.
With $k \ge 2^{d+1}+3$,
this forces at least $2^{d+1}+1$ crossings in $H$, which has at most $2^{d+1}$ edges, a contradiction.
\qed\end{proof}

The following four lemmas (\Cref{lemma:parallel,lemma:clockwise_characterization,lemma:2fusion,lemma:bounded_non_ab_outerplanar})
will enable us to show the safety of \Cref{rule:II}, defined in \Cref{section:phase_1}.
Bannister, Cabello, and Eppstein showed the following for all 1-planar embeddings. Hence it also applies, in particular, to geometric drawings. We restate it as a lemma using our terminology.
More concretely, they showed that if $G$ as stated in the following lemma admits a 1-planar embedding, $a$ and $b$ lie in a shared region.
\begin{lemma}[\textnormal{Adapted from \protect\cite[Lemma~6, Case $|S|=2$]{1planar_parameterized}}]\label{lemma:parallel}
    Let $G_1,\ldots,G_k$ be connected graphs on pairwise-disjoint vertex sets, 
    except that the only vertices allowed to appear in more than one $G_i$ are two
    distinct vertices $a,b$.
    Furthermore, let $G \coloneqq \bigcup_{i=1}^{k} G_i$, and let $d$ be its treedepth.
    If $k > 2^d$ and $G$ is geometric 1-planar,
    then $G$, as well as all $G_i$, are $(a,b)$-shared geometric 1-planar.
\end{lemma}

Before proving \Cref{lemma:2fusion}, we record a concise consequence of Thomassen’s characterization.
In any 1-planar embedding formed from \((a,b)\)-crossing pairs (and possibly the edge \(ab\)), being free of \(B\)- and \(W\)-configurations is equivalent to the following: 
the edges at \(a\) in clockwise order coming from \((a,b)\)-crossing pairs (and possibly \(ab\)) appear around \(a\) in the pattern \(L^{*}[M]R^{*}\) (all left crossings, optional ``middle'' edge \(ab\), then all right crossings); see \Cref{figure:clockwise_characterization}.
\begin{lemma}\label{lemma:clockwise_characterization}
Let $\varepsilon$ be a 1-planar embedding obtained from the union of two isolated vertices $a,b$, optionally edge $ab$, and a sequence of $(a,b)$-crossing pairs. 

Let $\pi_a$ denote the clockwise cyclic order of edges incident to $a$, read from the first edge after the outer region at $a$.  
Encode $\pi_a$ as a word over the alphabet $\{L,R,M\}$ by writing
\begin{itemize}
    \item $L$ for each incident edge from an $(a,b)$-left-crossing pair,
    \item $R$ for each incident edge from an $(a,b)$-right-crossing pair, and
    \item $M$ for edge $ab$.
\end{itemize}

Then $\varepsilon$ is free of $B$- and $W$-configurations if and only if 
all $L$'s in $\pi_a$ appear consecutively, followed (if present) by a single $M$, 
and then all $R$'s consecutively, i.e., $\pi_a$ can be written as $L^*[M]R^*$, where $\cdot^*$ denotes zero or more occurrences, and $[\cdot]$ denotes at most one occurrence.
\end{lemma}
\begin{proof}
    In the following, a subword is in general not contiguous. Observe that, by construction, each subword of $\pi_a$ of length two corresponds to a bounded region adjacent to $a$ in $\varepsilon$.
    We prove the contrapositive of the statement.

    $(\Rightarrow):$
    Suppose $\varepsilon$ contains a $B$ or $W$ configuration.
    If it is a $B$, $\pi_a$ contains $RM$ or $ML$ as a subword, depending on whether the crossed edges are $(a,b)$-right or $(a,b)$-left crossing.
    Similarly, if it is a $W$, $\pi_a$ contains $RL$.
    In all cases, $\pi_a$ does not adhere to the required pattern.

    $(\Leftarrow):$
    If $\pi_a$ deviates from the pattern $L^*[M]R^*$, it contains as a subword either $RL$, $RM$, $ML$, or $MM$.
    Subword $MM$ is impossible since edge $ab$ appears at most once in $\varepsilon$, $RL$ induces a $W$ configuration in $\varepsilon$, and $RM, ML$ both induce a $B$ configuration in $\varepsilon$.
\qed\end{proof}

    \begin{figure}[t]
    \centering
    \includegraphics[page=9]{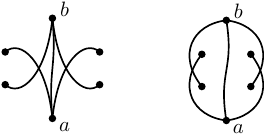}
\caption{Example for \Cref{lemma:clockwise_characterization}.
Left: Embedding with crossing sequence $LMR$, free of $B$- and $W$-configurations.
Right: Embedding with crossing sequence $RML$, containing two $B$- and one $W$-configuration.}
    \label{figure:clockwise_characterization}
\end{figure}

We are now ready to prove \Cref{lemma:2fusion}, which together with \Cref{lemma:parallel} will justify the correctness of the deletion step in \Cref{rule:II}.
\begin{lemma}\label{lemma:2fusion}
Let $G_1, G_2$ be connected graphs with disjoint vertex sets except for two distinct vertices $a,b$ that belong to both $G_1$ and $G_2$. Furthermore, edge $ab$ belongs to at most one of $G_1,G_2$. If $G_1$ is $(a,b)$-shared geometric 1-planar and $G_2$ is $(a,b)$-outer geometric 1-planar, then $G_1\cup G_2$ is $(a,b)$-shared geometric 1-planar.
\end{lemma}
\begin{proof}
Let $\varepsilon_1$ be a geometric 1-planar embedding of $G_1$ in which $a$ and $b$ share a region, and let $\varepsilon_2$ be a geometric 1-planar embedding of $G_2$ in which $a$ and $b$ lie on the outer region.

\begin{figure}[t]
    \centering
    \includegraphics[page=8]{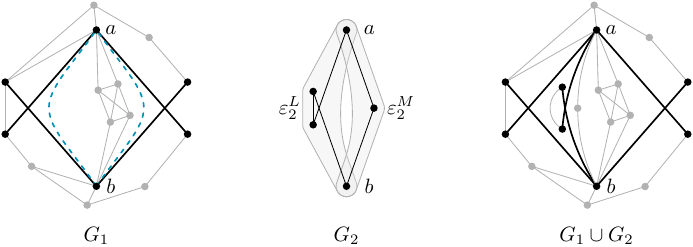}
\caption{Example for \Cref{lemma:2fusion}.
Left: Embedding $\varepsilon_1$ of $G_1$ shown in gray, with subembedding $\varepsilon_1'$ in bold black and lines of sight in dashed blue.
Middle: Embedding $\varepsilon_2$ of $G_2$, with subembeddings $\varepsilon_2^L$ and $\varepsilon_2^M$ indicated by gray regions (subembedding $\varepsilon_2^R$ is empty).
Right: Embedding $\varepsilon$ of $G_1 \cup G_2$, with subembedding $\varepsilon'$ in bold black.}
    \label{figure:fusion}
\end{figure}

Let $\varepsilon_1'$ be the restriction of $\varepsilon_1$ to all edges that belong to an $(a,b)$-crossing pair, together with $ab$ if present. Since $\varepsilon_1$ is geometric, it is
free of $B$- and $W$-configurations by Thomassen~\cite{straightening_characterization}, and so is $\varepsilon_1'$. Thus \Cref{lemma:clockwise_characterization} applies to $\varepsilon_1'$ and the clockwise order $\pi$ of $(a,b)$-pairs around $a$, represented as a word as defined in \Cref{lemma:clockwise_characterization}, has the form $L^{*}[M]R^{*}$.

\medskip
\noindent
\emph{Lines of sight in $\varepsilon_1'$.}
For each $(a,b)$-\emph{left}-crossing pair $\{aa',bb'\}$ with crossing point $c$, there is a line of sight from $a$ to $b$ \emph{immediately clockwise after} $aa'$ at $a$, following the concatenation of the subsegment of $aa'$ from $a$ to $c$ with the subsegment of $bb'$ from $c$ to $b$. For each $(a,b)$-\emph{right}-crossing pair, there is a line of sight \emph{immediately clockwise before} its edge at $a$. If $ab$ is present, there are lines of sight \emph{immediately before} and \emph{immediately after} $ab$. Observe that these lines of sight carry over to $\varepsilon_1$. See the left part of \Cref{figure:fusion}.

We now partition $\varepsilon_2$ into subdrawings $\varepsilon_2^L, \varepsilon_2^M, \varepsilon_2^R$.
Mark each vertex of $(a,b)$-left crossing edges in $\varepsilon_2$ as ``left'', and each vertex of $(a,b)$-right crossing edges as ``right''.
Assign an edge $e$ to subdrawing $\varepsilon_2^L$ if there is a path from an endpoint of $e$ not using $a,b$ to a vertex marked ``left''.
Define subdrawing $\varepsilon_2^R$ symmetrically.
Assign all remaining edges to subdrawing $\varepsilon_2^M$.
This is well-defined,
as an edge assigned to both $\varepsilon_2^L$ and $\varepsilon_2^R$ would imply there is a path from a ``left'' vertex to a ``right'' vertex not using $a,b$, which in turn would imply $a,b$ are not on the outer region in $\varepsilon_2$, violating the precondition that they are. See the middle part of \Cref{figure:fusion}.

\medskip
\noindent\emph{Insertion.}
We now create the (not necessarily geometric) embedding $\varepsilon$ of $G_1 \cup G_2$ by starting from $\varepsilon_1$ and inserting $\varepsilon_2^L, \varepsilon_2^M, \varepsilon_2^R$ at appropriate lines of sight between $a$ and $b$.
Note that when we insert a drawing in the following, we do so in a purely topological manner, i.e., the inserted drawing is in general no longer geometric.

If $\pi$ is empty, we use that $a$ and $b$ share a region in $\varepsilon_1$ and insert $\varepsilon_2$ into this region (we don’t need the decomposed drawing in this case).
If $\pi$ contains $M$, we insert at the line of sight directly \emph{preceding} $M$ $\varepsilon_2^L$, $\varepsilon_2^M$ in order, and at the line of sight \emph{succeeding} $M$ $\varepsilon_2^R$.
Otherwise, we select a line of sight after all $L$ symbols but before all $R$ symbols,
and insert $\varepsilon_2^L, \varepsilon_2^M, \varepsilon_2^R$ in order.
Since $a$ and $b$ lie on the outer region in $\varepsilon_2$ (and thus also in $\varepsilon_2^L, \varepsilon_2^M, \varepsilon_2^R$), this process does not introduce any new crossing. Hence, $\varepsilon$ inherits 1-planarity from $\varepsilon_1, \varepsilon_2$, and $a,b$ still share at least one region. Note also that all edges inserted into the drawing were not already present in $\varepsilon_1$ by precondition. See the right part of \Cref{figure:fusion}.

Let $\varepsilon'$ be the restriction of $\varepsilon$ to all $(a,b)$-crossing pairs together with $ab$ if present. By construction, the clockwise order of $(a,b)$-pairs around $a$ in $\varepsilon'$ in the sense of \Cref{lemma:clockwise_characterization} is again $L^{*}[M]R^{*}$. Hence, by \Cref{lemma:clockwise_characterization}, $\varepsilon'$ is free of $B$- and $W$-configurations.

If $\varepsilon$ contained a $B$ or $W$ configuration, it would have to contain edges from both $\varepsilon_1$ and $\varepsilon_2$, since $\varepsilon_1$ and $\varepsilon_2$ are geometric and thus free of $B$- and $W$-configurations. With $V(G_1)\cap V(G_2)=\{a,b\}$, any mixed $W$ consists of two $(a,b)$-crossing pairs, and any mixed $B$ consists of one such pair together with $ab$; in either case the configuration lies in $\varepsilon'$, contradicting that $\varepsilon'$ is free of $B$- and $W$-configurations. Therefore $\varepsilon$ is free of $B$- and $W$-configurations.

By Thomassen~\cite{straightening_characterization}, $\varepsilon$ is topologically equivalent to a geometric 1-planar embedding (straight-line edges with the same rotations at vertices and crossings). The region shared by $a$ and $b$ persists under this transformation, so $G_1\cup G_2$ is $(a,b)$-shared geometric 1-planar.
\qed\end{proof}

Finally, we derive a lemma to justify the rejection case in \Cref{rule:II}.
\begin{lemma}\label{lemma:bounded_non_ab_outerplanar}
    Let $G_1, \ldots, G_k$ be connected graphs with
    disjoint vertex sets, except two distinct vertices $a,b$ part of each $G_i$.
    Furthermore, each $G_i$ is not $(a,b)$-outer geometric 1-planar and consists of at most $m$ edges.
    Then, if $k \ge 2m+3$, graph $\bigcup_{i=1}^{k} G_i$ is not geometric 1-planar.
\end{lemma}
\begin{proof}
Assume, for the sake of contradiction, that there exists a geometric 1-planar embedding $\varepsilon$ of $\bigcup_{i=1}^{k} G_i$.

Let $\varepsilon_i$ denote the restriction of $\varepsilon$ to $G_i$, for each $i \in [k]$.  
By assumption, in $\varepsilon_i$, it is not possible for both $a$ and $b$ to lie on the outer region.  
Hence, there are at least $\lceil \tfrac{k}{2} \rceil$ embeddings $\varepsilon_i$ where, without loss of generality, $a$ is not on the outer region.  
Let $I \subseteq [k]$ index this subset.

Fix one such embedding $\varepsilon_j$, and consider the ``lasso'' shape obtained by following, in the planarization of $\varepsilon_j$, a path from $a$ to the outer face and then traversing the cycle around the outer face.

In the lasso shape of $\varepsilon_j$, segments corresponding to ``half-edges'' are already crossed with respect to $\varepsilon_j$.  
There remain at most $m$ segments corresponding to edges that are uncrossed in $\varepsilon_j$.

\begin{figure}[t]
    \centering
    \includegraphics[page=10]{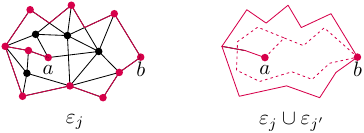}
    \caption{Illustration for \Cref{lemma:bounded_non_ab_outerplanar}. 
    Left: Embedding $\varepsilon_j$ with the lasso shape marked in red. 
    Right: The lasso shapes of $\varepsilon_j$ and $\varepsilon_{j'}$ (drawn dashed) necessarily intersect.}
    \label{figure:bounded_non_outerplanar}
\end{figure}

But observe that the lasso shape formed by each $\varepsilon_{j'}$ for $j' \in I \setminus \{j\}$ must cross the lasso of $\varepsilon_j$ at least once (see \Cref{figure:bounded_non_outerplanar}).  
Since $|I \setminus \{j\}| \ge \lceil k/2\rceil - 1$ and $k \ge 2m+3$, it follows that $\lceil k/2\rceil - 1 \ge m+1 > m$.  
Therefore, some edge of $\varepsilon_j$ is crossed at least twice in $\varepsilon$, contradicting 1-planarity.
\qed
\end{proof}

\subsection{Phase~I: Bounding the Maximum Block Size}
\label{section:phase_1}

Fix a connected graph $G$.
A treedepth decomposition of $G$ is a rooted forest on $V(G)$ in which, for every edge $xy\in E(G)$, one of $x,y$ is an ancestor of the other; the \emph{depth} of any rooted tree (and of any rooted subtree) is the maximum number of vertices on a root–leaf path.
Since $G$ is connected, we use a single rooted tree $T$ of depth $d$ on $V(G)$.
For $v\in V(T)$, let $\anc(v)$ be the set of \emph{ancestors} of $v$ in $T$ (including $v$), and let $\desc(v)$, the \emph{descendants} of $v$, be the vertices in the subtree rooted at $v$ (including $v$).
Write $G_v:=G[\desc(v)]$, and let $\anc(G_v)$ be the vertices outside $\desc(v)$ that are ancestors of at least one vertex of $\desc(v)$.
We may assume without loss of generality that every child subtree of a node has a neighbor in $\anc(v)$ (otherwise lift that child to be a sibling, which does not increase depth) and that each child subtree induces a connected subgraph of $G$ (otherwise split it into separate children, which does not increase depth as well).

\begin{ruleenv}\label{rule:I}
    Let $v\in V(T)$.
    If there exists $X\subseteq \anc(v)$ with $|X|\ge 3$ such that
    \[
      \bigl|\{\,c\text{ child of }v \mid \anc(G_c)=X\,\}\bigr| \;\geq\; 2^{\,d+1}+3,
    \]
    then reject the instance.
\end{ruleenv}

\begin{ruleenv}\label{rule:II}
    Let $v\in V(T)$,
    distinct $a, b \in \anc(v)$ that are part of a shared block in $G$. Set
    \[
      C\;:=\;\{\,c\text{ child of }v \mid \anc(G_c)=\{a,b\}\,\}.
    \]
    If $|C| \leq 2^{d}+1$, do nothing.  
    Otherwise, provided that neither 
    \Cref{rule:I}  nor \Cref{rule:II} applies to any proper descendant of $v$:
    choose an overflow set $O\subseteq C$ with $|C\setminus O|=2^{d}+1$;  
    for each $c\in O$, test whether $G_c$ is $(a,b)$-outer geometric 1-planar (by brute force) and, if so, delete $G_c$ from $G$ and remove the subtree of $T$ induced by $\desc(c)$;  
    let $O'$ denote the remaining elements of $O$. 
    Write $m:=\max_{c\in O'}|E(G_c)|$.  
    If $|O'|\ge 2m+3$, reject the instance.
\end{ruleenv}

From the above lemmas, we directly obtain:
\begin{lemma}\label{lemma:td_correctness_rule_1_and_2}
\Cref{rule:I} and \Cref{rule:II} are safe.
\end{lemma}
\begin{proof}
If \Cref{rule:I} triggers for some $X\subseteq \anc(v)$ with $|X|\ge3$, then \Cref{lemma:bounded_claws} implies that the graph is not 1-planar and thus rejection is correct.

Next consider \Cref{rule:II}.
Let $C,O,O'$ be as in the rule for the fixed pair $\{a,b\}\subseteq \anc(v)$.
By construction we keep $|C\setminus O|=2^{d}+1$ children with attachment $\{a,b\}$ untouched.
Hence, after any deletions from $O$, there remain strictly more than $2^{d}$ siblings with attachment $\{a,b\}$.
Therefore, if the reduced instance is geometric 1-planar, \Cref{lemma:parallel} applies and forces $a$ and $b$ to share a region in the whole embedding. Given such an embedding, every deleted $(a,b)$-outer geometric 1-planar child can be reinserted by \Cref{lemma:2fusion}. Hence, deleting them was safe.
If, after deletions, $|O'|\ge 2m+3$ with $m=\max_{c\in O'}|E(G_c)|$, then \Cref{lemma:bounded_non_ab_outerplanar} forbids a geometric 1-planar embedding, so rejection is correct.
\qed
\end{proof}

\begin{lemma}\label{lemma:td_block_size}
After exhaustively applying \Cref{rule:I,rule:II}, let $N_\ell$ denote the maximum number of vertices in 
subtrees of the reduced treedepth decomposition of height $\ell$ that induce a biconnected graph. Then
\[
N_\ell \;=\; 2^{\,\mathcal{O}\big((d+\ell)\,3^{\ell}\big)} .
\]
\end{lemma}
\begin{proof}
Let $v$ be a node where $G_v$ is biconnected, and let $A:=\anc(v)$, so $|A|\le \ell+1$.  Each child $c$ of $v$ satisfies $\anc(G_c)=X\subseteq A$ and
has height at most $\ell-1$, hence $|V(G_c)|\le N_{\ell-1}$ and
$|E(G_c)|\le\binom{N_{\ell-1}}{2}$.

\emph{Sets $X$ with $|X|\ge 3$.}
By \Cref{rule:I}, for every such $X$ there are at most $2^{d+1}+2$ children with $\anc(G_c)=X$.
There are at most $2^{\ell+1}$ choices of $X$, so these contribute
$O(2^{\ell+d}\,N_{\ell-1})$ vertices in total.

\emph{Sets $X$ with $|X|=2$.}
There are $\binom{\ell+1}{2}=O(\ell^2)$ pairs $\{a,b\}\subseteq A$.
Fix one pair and let $C:=\{c:\anc(G_c)=\{a,b\}\}$.
If $|C|\le 2^{d}+1$, nothing is removed.
Otherwise \Cref{rule:II} chooses an overflow set $O\subseteq C$ with $|C\setminus O|=2^{d}+1$,
deletes those $G_c$ in $O$ that are $(a,b)$-outer geometric 1-planar, and lets $O'$ be the remainder.
If $|O'|\ge 2m+3$ (where $m:=\max_{c\in O'}|E(G_c)|$), the instance would be rejected; hence $|O'|\le 2m+2$.
Using $m\le\binom{N_{\ell-1}}{2}$, the number of surviving children for this pair is at most
\[
(2^{d}+1)+(2m+2) \;\le\; 2^{d} + \big(N_{\ell-1}^2 - N_{\ell-1} + 3\big).
\]
Each contributes at most $N_{\ell-1}$ vertices, so over all pairs the contribution is
$O\!\big(\ell^2 N_{\ell-1}^3\big) + O\!\big(\ell^2 2^{d} N_{\ell-1}\big)$.

Adding the vertex $v$ and combining both cases,
\[
N_\ell \;\le\; 1 \;+\; O\!\big(\ell^2 N_{\ell-1}^3\big) \;+\; O\!\big(2^{\ell+d} N_{\ell-1}\big)
\;\le\; 2^{\,O(\ell+d)}\,N_{\ell-1}^{3}\quad(\ell\ge1,\ N_{\ell-1}\ge2).
\]
Taking $\log_2$ and unrolling,
\[
\begin{aligned}
\makebox[\linewidth][l]{$\log_2 N_\ell \le 3\log_2 N_{\ell-1} + O(\ell+d) \;\Rightarrow$}\\
\makebox[\linewidth][r]{$\log_2 N_\ell \le \sum_{i=1}^{\ell} 3^{\ell-i}\,O(d+i)
= O\!\big((d+\ell)\,3^{\ell}\big)$}
\end{aligned}
\]
Therefore $N_\ell=2^{\,\mathcal{O}((d+\ell)\,3^{\ell})}$.
\qed
\end{proof}

\subsection{Preliminaries for Phase~II}
\label{section:prelims_phase_2}

The next two lemmas, mirroring \Cref{lemma:2fusion,lemma:bounded_non_ab_outerplanar}, ensure that \Cref{rule:III}, defined in \Cref{section:phase_2}, is safe.
\begin{lemma}\label{lemma:1fusion}
Let $G_1, G_2$ be connected graphs with disjoint vertex sets except for vertex $a$ belonging to both $G_1$ and $G_2$. If $G_1$ is geometric 1-planar and $G_2$ is $a$-outer geometric 1-planar, then $G_1\cup G_2$ is geometric 1-planar.
\end{lemma}
\begin{proof}
Take a geometric 1-planar embedding $\varepsilon_1$ of $G_1$,
and a geometric 1-planar embedding of $G_2$ where $a$ is on the outer region.
Clearly, one can (whilst losing straightness)
attach embedding $\varepsilon_2$ at $a$ in $\varepsilon_1$, so that both embeddings, in a topological sense, are unchanged with respect to themselves, and no new crossings are introduced.
Call this 1-planar embedding of $G_1 \cup G_2$  $\varepsilon$.
Now, suppose there is a $B$- or $W$-configuration in $\varepsilon$. 
The planarization of the $B$- or $W$-configuration appears as a subgraph in the planarization of $\varepsilon$.
It cannot be contained in $\varepsilon_1$ or $\varepsilon_2$ alone, since both embeddings are free of $B$- and $W$-configurations.
Dummy vertices of the embedded configuration must appear in either $\varepsilon_1$ or $\varepsilon_2$ as we did not create new crossings.
But then we obtain that the cycle occurring in the planarized configuration touches vertices that are unique to $G_1$, and also vertices that are unique to $G_2$.
This is impossible since $a$ is a cut-vertex of $G_1 \cup G_2$.
Hence, by Thomassen~\cite{straightening_characterization}, $\varepsilon$ can be straightened to yield the required embedding. \qed
\end{proof}

\begin{lemma}\label{lemma:bounded_non_a_outerplanar}
Let $G_1,\ldots,G_k$ be connected graphs with pairwise-disjoint vertex sets except that $a$ belongs to every $G_i$.
Assume each $G_i$ is not $a$-outer geometric 1-planar and has at most $m$ edges.
If $k\ge m+2$, then $\bigcup_{i=1}^{k}G_i$ is not geometric 1-planar.
\end{lemma}
\begin{keepproof}
Apply the lasso argument of \Cref{lemma:bounded_non_ab_outerplanar}, but simplified in the sense that we can assume all lassos are attached to $a$.
\end{keepproof}

\subsection{Phase~II: Bounding the Number of Blocks}
\label{section:phase_2}
Consider the block-cut tree of $G$, rooted at an arbitrary cut vertex of $G$, after exhaustively applying \Cref{rule:I} and \Cref{rule:II}.
\begin{ruleenv}\label{rule:III}
    Let $v$ be a cut vertex of the block-cut tree.
    If all descendant cut vertices of $v$ are already processed, the rule is applicable. 
    Let $C$ be the set of children of $v$ in the block-cut tree, and for each $c \in C$, let $T_c$ denote the sub-block-cut tree rooted at $c$. We set $G_c$ to be the induced subgraph of $G$ obtained by taking the union of all blocks of $T_c$.
    For each $c \in C$, using brute force, check whether $G_c$ is $v$-outer geometric 1-planar. If it is, delete $G_c$ from the instance and $c$ from the block-cut tree.
    Let $C' \subseteq C$ denote the set of children not deleted.
    Write $m:=\max_{c\in C'}|E(G_c)|$.  
    If $|C'|\ge m+2$, reject the instance.
\end{ruleenv}
\Cref{lemma:1fusion,lemma:bounded_non_a_outerplanar} directly imply:
\begin{lemma}\label{lemma:td_correctness_rule_3}
\Cref{rule:III} is safe.
\end{lemma}
\begin{proof}
Let $v$ be the processed cut vertex, $C$ its children in the block–cut tree, and $G_c$ the subgraph for $c\in C$ as in \Cref{rule:III}.
If $G_c$ is $v$-outer geometric 1-planar and we delete it, the deletion is safe: whenever the remaining instance is geometric 1-planar, we can reinsert $G_c$ at $v$ by \Cref{lemma:1fusion}.
If, after deletions, a set $C'\subseteq C$ remains with $|C'|\ge m+2$ where $m=\max_{c\in C'}|E(G_c)|$, then the graphs $\{G_c\}_{c\in C'}$ are connected, pairwise vertex-disjoint except for $v$, none is $v$-outer geometric 1-planar, and each has at most $m$ edges; by \Cref{lemma:bounded_non_a_outerplanar} their union is not geometric 1-planar, so rejection is correct.
\qed
\end{proof}

\begin{lemma}\label{lemma:bounded_block_cut_tree}
After exhaustively applying \Cref{rule:I,rule:II,rule:III}, the total number of blocks is at most
\[
\mathcal{O}\left(
\bigl(2\tbinom{N_d}{2}\bigr)^{\,2^{\mathcal{O}(2^{d})}} \right),
\]
where $N_d$ is the bound from \Cref{lemma:td_block_size} for subtrees of the treedepth decomposition inducing a biconnected graph at height $d$.
\end{lemma}
\begin{proof}
Let $B_\ell$ be the maximum number of blocks in any subtree of the block-cut tree of height $\ell$, with $B_0=1$.
Each child $c$ contributes at most $B_{\ell-1}$ blocks; by \Cref{lemma:td_block_size} every block has at most $N_d$ vertices, so $|E(G_c)|\le \binom{N_d}{2}\,B_{\ell-1}$.
\Cref{rule:III} bounds the number of children by this edge bound plus one, hence
\[
B_\ell \;\le\; \bigl(\tbinom{N_d}{2}\,B_{\ell-1}+1\bigr)\,B_{\ell-1}
\;\le\; \bigl(2\tbinom{N_d}{2}\bigr)\,B_{\ell-1}^2,
\]
and by induction $B_\ell \le \bigl(2\binom{N_d}{2}\bigr)^{2^{\ell}-1}$.
Finally, every path in the block–cut tree of length $k$ corresponds to a simple path in $G$ of length $\Omega(k)$, and treedepth $d$ bounds simple-path length in $G$ by $<2^{d}$; thus the block–cut tree has height  $\mathcal{O}(2^{d})$, and substituting $\ell=\mathcal{O}(2^{d})$ yields the claimed bound.
\qed
\end{proof}

Finally, we have all ingredients to derive \Cref{theorem:fpt_treedepth}. 

\begin{theoremreused}{theorem:fpt_treedepth}
Let $G$ be a graph on $n$ vertices with treedepth at most $d$.
Then \textsc{Geometric 1-Planarity} can be decided in time
\[
\mathcal{O}\!\bigl(2^{\,2^{\,2^{\,2^{\mathcal{O}(d)}}}}\cdot n^{\mathcal{O}(1)}\bigr).
\]
\end{theoremreused}

\begin{proof}
We can process each connected component of $G$ independently with polynomial overhead, thus we assume $G$ is connected.
Let $T$ be a treedepth decomposition of $G$ of depth $d$.
Apply \Cref{rule:I} and \Cref{rule:II} exhaustively in a bottom–up traversal of $T$.
By \Cref{lemma:td_block_size}, every block has at most
\[
N_d \;=\; 2^{\mathcal{O}(d\,3^{d})}\;\le\;2^{\,2^{\mathcal{O}(d)}}
\]
vertices.

Build the block–cut tree and apply \Cref{rule:III} exhaustively.
By \Cref{lemma:bounded_block_cut_tree}, the total number of blocks is at most
\[
\bigl(2\textstyle\binom{N_d}{2}\bigr)^{\,2^{\mathcal{O}(2^{d})}}.
\]
Hence the resulting graph has at most
\[
S(d)\;\coloneqq\;N_d\cdot\bigl(2\textstyle\binom{N_d}{2}\bigr)^{\,2^{\mathcal{O}(2^{d})}}
\;\le\;2^{\,2^{\mathcal{O}(2^{d})}}
\]
vertices.

Whenever a reduction rule requires deciding a predicate (namely, ``$(a,b)$-outer geometric 1-planar'' in \Cref{rule:II} and ``$v$-outer geometric 1-planar'' in \Cref{rule:III}), that predicate is in \NP{} via Thomassen’s characterization~\cite{straightening_characterization}, and by $\NP \subseteq \EXP$, it can be decided in time $2^{\,q^{\mathcal{O}(1)}}$ on inputs with $q$ vertices.
Here $q\le S(d)$, and the total number of such predicate evaluations is polynomial in $n$.

After all rules stop, the remaining instance has at most $S(d)$ vertices.
Since \textsc{Geometric 1-Planarity} is in \NP{}, the final instance can be decided in time $2^{\,S(d)^{\mathcal{O}(1)}}$.
Using $S(d)\le 2^{\,2^{\mathcal{O}(2^{d})}}$ yields the claimed overall running time.
\qed
\end{proof}

\section{Kernelization via Feedback Edge Number}
\label{section:fen}
\iflong
We begin this section by deriving \Cref{theorem:fen_fpt} and \Cref{corollary:fen_k_planarity}.
\else
In this section, we derive \Cref{theorem:fen_fpt} and \Cref{corollary:fen_k_planarity}.
We refer to the appendix for the proof of \Cref{theorem:geom_k_planar_fen_fpt}, i.e., the kernel for \textsc{Geometric $k$-Planarity}.
\fi
Let $G$ be a graph with feedback edge number $\ell \geq 1$. (The problem is trivial for $\ell = 0$.) We remark that it is folklore that an optimal feedback edge set of $G$, i.e., a set $F \subseteq E(G)$ with $|F| = \ell$ whose removal makes $G$ acyclic, can be obtained in polynomial time by computing a spanning tree.
Without loss of generality, we may assume that $G$ has no degree-1 vertices, since iteratively deleting such vertices clearly does not affect (geometric) 1-planarity.
In a graph $H$, a \emph{degree-2 path} is a simple path whose internal vertices all have degree two in $H$; note that this also includes paths of length one.
We define the kernel $G'$ for \textsc{1-Planarity} and $\overline{G'}$ for \textsc{Geometric 1-Planarity} as follows.
By \cite[Lemma~8]{1planar_parameterized}, the edge set of $G$ can be decomposed into at most $3\ell-3$ maximal degree-2 paths. 
Sort these paths by increasing length as $P_1, P_2, \dots, P_p$, where $p \le 3\ell-3$.
For $i > 1$, we call $P_i$ \emph{long} if $|E(P_i)| \ge p - 1 + \sum_{j=1}^{i-1} |E(P_j)|$, and \emph{very long} if $|E(P_i)| \ge 2(p - 1 + \sum_{j=1}^{i-1} |E(P_j)|)$.
We first describe how to obtain $G'$, the kernel for \textsc{1-Planarity}.  
If $P_1$ has length at least $p - 1$, set $G' \coloneqq K_2$, a trivial yes-instance.  
If no $P_i$ is long, set $G' \coloneqq G$.  
Otherwise, let $j$ be the smallest index such that $P_j$ is long.  
Then, construct $G'$ as the union of all $P_i$ with $i < j$, together with all $P_i$ for $i \ge j$, shortened to length $|E(P_j)|$.
The kernel $\overline{G'}$ for \textsc{Geometric 1-Planarity} is defined analogously, using the predicate ``very long'' in place of ``long.''

In the following lemma, we show that, by first applying Reidemeister moves of type~I and II~\cite{reidemeister} and then---in the geometric case---applying Thomassen's characterization~\cite{straightening_characterization}, one can transform a solution drawing of the original graph into a solution drawing of the kernel. In particular, a type~I move removes a self-crossing of a (very) long path, while a type~II move removes two crossings between two mutually crossing (very) long paths.
\begin{lemma}\label{lemma:global_crossing_elimination}
Let $G$ be a graph, and let $\varepsilon$ be a 1-planar (resp.\ geometric 1-planar) embedding of $G$.
Partition $E(G)$ into $s$ edges, which we call \emph{static edges}, and $f$ maximal degree-2 paths, which we call \emph{flexible paths}, each of length at least $s + f - 1$ (resp.\ $2\cdot(s + f - 1)$).
Let $G'$ be obtained from $G$ by shortening each flexible path to length $s + f - 1$ (resp.\ $2\cdot(s + f - 1)$) while preserving its endpoints.
Then there exists a 1-planar (resp.\ geometric 1-planar) embedding of $G'$.
\end{lemma}
\begin{proof}
We first prove the statement for the non-geometric case and then lift it to the geometric setting.

\begin{figure}[t]
    \centering
    \includegraphics[page=6]{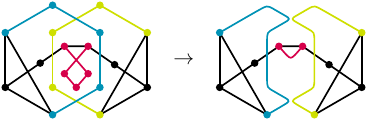}
    \caption{Illustration of the drawing simplification rules from \Cref{lemma:global_crossing_elimination}. The left side shows the original drawing, while on the right, flexible paths are represented as Jordan arcs. Static edges are shown in black, and flexible paths are colored. Rule~I is applied to the red flexible path, and Rule~II to the blue and green flexible paths.}
    \label{figure:global_crossing_elimination}
\end{figure}

\smallskip
\noindent\emph{Non-geometric case.}
We begin by redrawing the flexible paths.
For each flexible path, consider the curve obtained by concatenating the Jordan arcs corresponding to its edges.
We refer to the collection of these arcs as the set of \emph{flexible arcs}.
We apply the following two crossing-elimination rules exhaustively to the flexible arcs to reduce the total number of crossings between them, while preserving, for each static edge, whether it is crossed by a flexible arc.

\begin{enumerate}
    \item[\textbf{Rule~I}] Whenever a flexible arc crosses itself, shortcut the arc by removing the loop.
    \item[\textbf{Rule~II}] Whenever two flexible arcs cross each other twice, interchange the two subarcs enclosed between the crossings.
\end{enumerate}
We remark that Rule~I and Rule~II correspond to the Reidemeister moves of type I and type II, respectively, which are foundational to knot theory \cite{reidemeister}.

See \Cref{figure:global_crossing_elimination} for an example.
Rule~I decreases the total number of crossings by one, and Rule~II by two.
Hence, the process must terminate.
After termination, each flexible arc is simple (i.e., it has no self-crossings) and crosses any other flexible arc at most once.

In the worst case, each flexible arc may intersect every static edge once.
Consequently, each flexible arc participates in at most $s + f - 1$ crossings in total.

We now construct a 1-planar embedding $\varepsilon'$ of $G'$ as follows.
Retain the embedding induced by the static edges in $\varepsilon$.
Insert the simplified flexible arcs into this embedding.
Then, for each flexible arc, traverse it from start to end, and insert a subdivision vertex immediately after each crossing.
If any additional subdivision vertices remain unused, they can be placed arbitrarily.
This yields a valid 1-planar embedding of $G'$, completing the first part of the proof.

\smallskip
\noindent\emph{Geometric case.}
We proceed analogously, now starting from a geometric 1-planar embedding.
When resolving a crossing, we insert two subdivision vertices instead of one—one on each side of the crossing.
Let $\varepsilon'$ denote the resulting 1-planar embedding of $G'$.
We now argue that $\varepsilon'$ is free of $B$ and $W$ configurations, and can hence be straightened by Thomassen's characterization~\cite{straightening_characterization}.
Without loss of generality, assume we did not need to insert any ``leftover'' subdivision vertices, as these can safely be inserted after the straightening step.

It suffices to show that $\varepsilon'$ contains no $B$- or $W$-configuration that touches an internal vertex of a shortened flexible path.
Otherwise, such a configuration would consist entirely of static edges, contradicting the fact that the subdrawing induced by the static edges in the original embedding $\varepsilon$ was geometric 1-planar and therefore free of $B$- and $W$-configurations.

\begin{figure}[!htbp]
    \centering
    \includegraphics[page=3]{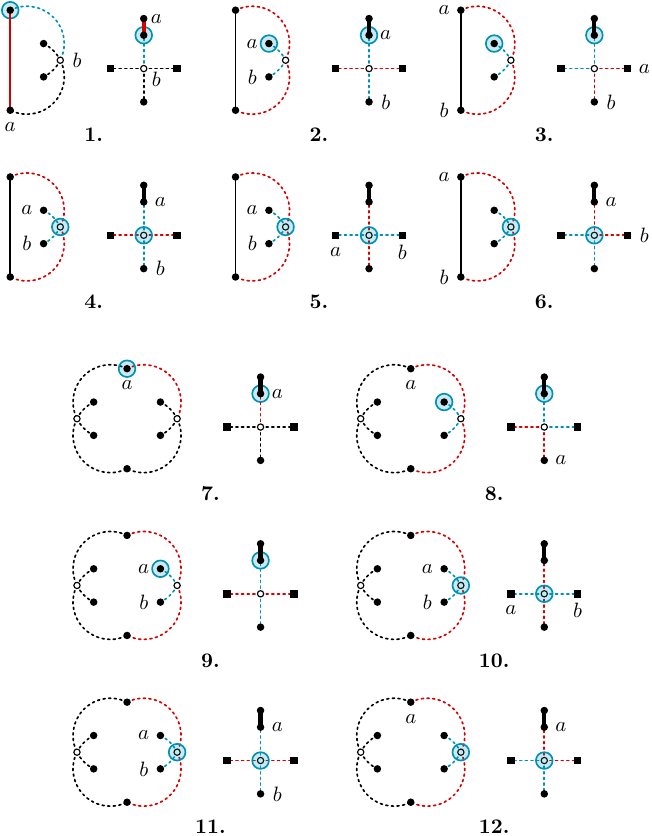}

    \caption{The 12 cases (up to symmetry) of trying to embed a $W/B$ configuration into the planarization of $\varepsilon'$ such that an internal vertex of a flexible path is touched.
    In each case, left is a $B$ or $W$ configuration, right an induced subgraph of the planarization of $\varepsilon'$, consisting of one real internal flexible path vertex adjacent to an uncrossed edge and a crossed edge, which is crossed by another edge not part of the path). The blue circles denote how a real or dummy vertex of a $B$- or $W$-configuration is mapped. Real vertices are drawn as filled disks, dummy vertices unfilled, the endpoints of the edge crossing the path with squares. Half-edges are dashed, uncrossed edges fat. The colors red/blue identify how the respective edges are mapped.}
    \label{figure:fen_no_BW_configs}
\end{figure}

We show this via case analysis; see \Cref{figure:fen_no_BW_configs}.
We consider the planarization of $\varepsilon'$,
and enumerate (up to symmetry) all ways the planarization of a $B$- or $W$-configuration can be embedded into the planarization of $\varepsilon'$ such that an internal vertex of a planarized flexible path is part of the planarized $B$- or $W$-configuration. 
In each of the 12 cases, we map the (dummy or real) vertex encircled in blue of a $B$- or $W$-configuration to one internal vertex of the respective type of a flexible path. Note that each real such vertex is, by construction, adjacent to one uncrossed edge and one half-edge.

In Case~1, the mapping of the blue half-edge adjacent to the blue vertex in the $B$-configuration is forced to one unique option, the blue half-edge with endpoint labeled $b$, as the blue vertex on the flexible path is adjacent to only one half-edge.
This forces the edge marked in red to be mapped to the uncrossed edge with endpoint labeled $a$.
In the $B$-configuration, $a$ and $b$ need to be connected via a half-edge, but they are not on the path. Hence, the chosen blue vertex cannot touch an internal vertex of a flexible path.

We proceed similarly for the remaining 11 cases.
In each case, we consider a vertex of a $B$- or $W$-configuration (real or dummy). In some cases, the mapping of the adjacent edges is forced as in Case 1, in others, there are multiple options, handled in separate cases.

We derive a contradiction in each case as follows:
In Cases 1, 3, and 6, the vertices $a$ and $b$ are not adjacent, but need to be in the $B$- or $W$-configuration. 
In Cases 2, 4, 5, and 9--11, both $a$ and $b$ must lie in the region bounded by the $B$- or $W$-configuration, 
but the mapping forces that only one is. 
Finally, in Cases 7, 8 and 12, the vertex $a$ should have two adjacent half-edges, 
but instead it has only one (together with one uncrossed edge).
\qed\end{proof}

We now justify the ``base case'' of our kernelization. Intuitively, if the shortest path is sufficiently long, we have a trivial yes-instance, as then, we can draw each path as a straight line with the path's endpoints in convex position.
\begin{lemma}\label{lemma:trivial_yes_instance}
Let $G$ be a graph partitioned into $f$ degree-2 paths, each of length at least $f-1$.
Then, $G$ is geometric 1-planar.
\end{lemma}
\begin{proof}
    To obtain a certifying drawing, place $G$'s vertices in convex position (e.g., on the boundary of a circle), and draw a straight line between the first and last vertex of each path.
    It remains to insert the subdivision vertices of each path to obtain a proper drawing of $G$.
    Observe that each line crosses every other line at most once. Hence, we have enough segments per path ($\geq f-1$) to tolerate the at most $f-1$ crossings, and can place the subdivision vertices accordingly.
\qed\end{proof}

Solving one recurrence per kernel allows us to bound their size. We analyze the worst case, where the shortest path is too short for trivial rejection and no path is (very) long, so no shortening occurs.
\begin{lemma}\label{lemma:fen_kernel_size}
    Graph $G'$ has at most $\mathcal{O}(\ell \cdot 8^{\ell})$ edges, and
    $\overline{G'}$ at most $\mathcal{O}(\ell \cdot 27^{\ell})$.
\end{lemma}
\begin{proof}
First, we bound $|E(G')|$.  
The worst case is achieved when each $P_i$ is one edge too short to be considered long (i.e., no shortening of long paths takes place), and the first path $P_1$ is one edge too short for the procedure to replace $G$ with a trivial yes-instance. 
Then, we have 
\begin{align*}
    |E(P_i)| &= p - 2 + \sum_{j=1}^{i-1} |E(P_j)| \quad \text{for } 2 \leq i \leq p, \\
    |E(P_1)| &= p - 2 .
\end{align*}
Set 
\[
    S_i \coloneqq \sum_{j=1}^{i} |E(P_j)|,
\]
so that $S_p = |E(G')|$.  
Then, rearranging, we obtain
\begin{align*}
    S_i &= 2 S_{i-1} + p - 2 \quad \text{for } 2 \leq i \leq p, \\
    S_1 &= p - 2 .
\end{align*}
Solving this recurrence yields
\[
    S_i = (2^i - 1)(p - 2) \quad \text{for } 1 \leq i \leq p .
\]
As the number of paths $p$ is at most $3\ell-3$, we obtain
\[
    |E(G')| = S_p \leq \frac{(8^\ell - 8)(3\ell - 5)}{8} = \mathcal{O}(\ell \cdot 8^\ell).
\]

To bound $|E(\overline{G'})|$, we proceed symmetrically, but use the predicate ``very long path'' instead of ``long path''.
This way, we obtain
\begin{align*}
    |E(\overline{G'})|
      &\leq \tfrac{9}{2} \;-\; \tfrac{19}{2}\,3^{3\ell-4}
           \;+\; \left(-3 + 2 \cdot 27^{\,\ell-1}\right) \ell \\
      &= \mathcal{O}(\ell \cdot 27^\ell) \tag*{\qed}
\end{align*}
\end{proof}

Combining the above lemmas yields that our kernelization is correct:
\begin{lemma}\label{lemma:fen_correctness}
$G$ is (geometric) 1-planar if and only if $G'$ (resp. $\overline{G'}$) is.
\end{lemma}
\begin{proof}
$(\Rightarrow):$
Assume $G$ is 1-planar (resp. geometric 1-planar).
If $G' = K_2$, which if the case if $P_1$ has length at least $p - 1$, $G'$ is trivially 1-planar (resp. geometric 1-planar).
Next, if no $P_i$ is long (resp. very long), we have $G' = G$ (resp. $\overline{G'} = G$), and the statement holds trivially.
Otherwise, let $j$ be the minimum index of a long (resp. very long) path.
Then, we can apply \Cref{lemma:global_crossing_elimination} to $G$ with static edges $\bigcup_{i=1}^{j-1} E(P_i)$
and flexible paths $P_j, \dots, P_p$, and obtain that $G'$ is 1-planar (resp. geometric 1-planar).

$(\Leftarrow):$
Assume $G'$ is 1-planar (resp. $\overline{G'}$ is geometric 1-planar).
If $P_1$ has length at least $p - 1$,
all $p$ paths are at least this long. Hence, by \Cref{lemma:trivial_yes_instance}, $G$ is (geometric) 1-planar.
Next, if no $P_i$ is long (resp. very long), we have $G = G'$ (resp. $G = \overline{G'}$), and the statement holds trivially.
Otherwise, observe that we can obtain $G$ from $G'$ (resp. $\overline{G'}$) by adding subdivision vertices.
Since we can add subdivision vertices to a witness embedding of $G'$ (resp. $\overline{G'}$) while maintaining (geometric) 1-planarity, $G$ is 1-planar (resp. geometric 1-planar).
\qed\end{proof}

We now have all prerequisites to derive \Cref{theorem:fen_fpt}.
\textsc{1-Planarity} (and, more generally, \textsc{$k$-Planarity}) is trivially in~\NP.
\textsc{Geometric 1-Planarity} is also in~\NP, as shown by Hong et~al.~\cite{rediscovered}.
Hence all three problems are decidable.
In \Cref{lemma:fen_correctness} we proved the correctness of our kernelization, 
and in \Cref{lemma:fen_kernel_size}, we bounded the size of the resulting instances $G'$ and~$\overline{G'}$.
Thus, we have:

\fenfpt*

A graph~$G$ is $k$-planar if and only if the graph~$G_k$, obtained by replacing each edge of~$G$ with a path of~$k$ edges, is 1-planar~\cite{k_planarity_gd}.
Since $G_k$ has the same feedback edge number as~$G$, we immediately obtain the following.

\begin{corollary}\label{corollary:fen_k_planarity}
    \textsc{$k$-Planarity}, parameterized by the feedback edge number~$\ell$, admits a kernel with $\mathcal{O}(\ell \cdot 8^\ell)$ edges.
\end{corollary}

\iflong
\subsection{Geometric $k$-Planarity}

We now apply our technique to \textsc{Geometric $k$-Planarity}.
We define the kernel exactly as in the beginning of \Cref{section:fen}, with the only difference being the notion of a ``long path.'' (In what follows, all symbols—including the paths $P_1, \dots, P_p$ and the feedback edge number of the input graph~$\ell$—are defined as in the beginning of \Cref{section:fen}.)

Instead of the long (respectively, very long) paths defined there, we construct the kernel with respect to \emph{very, very long paths}.  
Let $s(i) \coloneqq \sum_{j=1}^{i-1} |E(P_j)|$ for $i \in [p]$.  
We say that $P_i$ for $i > 1$ is \emph{very, very long} if  
$
|E(P_i)| \geq (s(i)^2 + 3s(i) + 1) \cdot (2 + p - 1) + 1.
$

To complement this notion, we need a new redrawing strategy, which is to partition a solution drawing into triangles where the interior of each triangle contains only very, very long paths, and then redraw the very, very long paths, essentially as straight lines, inside each triangle.
Formally, a \emph{constrained triangulation} of a geometric planar embedding 
is a geometric super-embedding with the same vertex set where every face is a triangle~\cite{computational_geometry_book}.

First, we give a bound telling us how many triangles are required.

\begin{lemma}\label{lemma:triangulations}
Let $G$ be a $k$-plane geometric graph with $m$ edges, drawn inside a fixed triangular region $\Delta$, and let $P$ be the planarization of this drawing (including $\Delta$). Define $t(m)$ as the maximum number of triangles in a constrained triangulation of $P$. Then
\[
t(m) \le m^2 + 3m + 1.
\]
\end{lemma}
\begin{proof}
The quantity $t(m)$ is well-defined, as every plane geometric graph admits a constrained triangulation, see e.g.,~\cite{constrained_delaunay}.  
If $P$ has $N \ge 3$ vertices (comprising the original endpoints, crossing points, and the three corners of $\Delta$) then any maximally planar straight-line graph on these vertices has exactly $2N - 4$ faces (see, e.g.,~\cite{diestel}), and hence $2N - 5$ triangular faces.  
In a geometric drawing, each pair of edges crosses at most once. Hence, $N \le 2m + \binom{m}{2} + 3$.
Thus
\[
t(m) \le 2N - 5 
   \le 2\!\left(2m + \tbinom{m}{2} + 3\right) - 5 
   = m^2 + 3m + 1.
\]
\end{proof}

We are now ready to present the redrawing argument.
\begin{lemma}\label{lemma:triangle_redrawing}
Let $G$ be a graph, and let $\varepsilon$ be a geometric $k$-planar embedding of~$G$.
Partition $E(G)$ into $s$ edges, which we call \emph{static edges}, and $f$ maximal degree-2 paths, which we call \emph{flexible paths}, each of length at least $(s^2 + 3s + 1)\cdot(2 + f - 1) + 1$.
Let $G'$ be obtained from $G$ by shortening each flexible path to length $(s^2 + 3s + 1)\cdot(2 + f - 1) + 1$ while preserving its endpoints.
Then there exists a geometric $k$-planar embedding of $G'$.
\end{lemma}
\begin{proof}
Fix a triangulation in the sense of \Cref{lemma:triangulations} of $\varepsilon$ restricted to the static edges with an added triangle bounding all static and flexible edges.
To obtain the desired drawing of $G'$, we will redraw $\varepsilon$ inside each triangle of the triangulation. Note that by construction, the interior of each triangle can only be intersected by edges of flexible paths.

\begin{figure}[t]
    \centering
    \includegraphics[page=2]{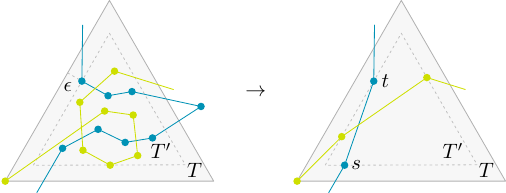}
    \caption{Example of the redrawing step of \Cref{lemma:triangle_redrawing}. Flexible paths are drawn in blue and green, triangulation triangle $T$ in gray, ``shrunken triangle'' $T'$ in dashed outline. On the right, $s$ and $t$ are labeled only for the blue path.}
    \label{figure:triangle_simplification}
\end{figure}

Order the triangles arbitrarily. 
For each triangle $T$, do the following (see \Cref{figure:triangle_simplification} for an example):
Let $\epsilon$ be the minimum distance from the boundary of $T$ to a vertex or crossing in the interior of $T$.
If there is no vertex or crossing inside $T$, we set $\epsilon \coloneqq \infty$ and proceed to the next triangle.
Otherwise, let $T'$ be the region obtained from $T$ by moving its boundary inward along the normal direction by distance $\epsilon$.

We redraw each flexible path $P$ that intersects with $T'$ as follows.
Traverse $P$ from one endpoint to the other. Let $s$ be the first intersection point with $T'$, and $t$ be the last intersection point with $T'$.
Subdivide $P$ at points $s$ and $t$, delete the segments between them, and join $s$ and $t$ via a straight line.

Next, we insert subdivision vertices inside $T'$ to ensure the drawing, restricted to $T'$, is 1-planar.
Since $T'$ is convex and each flexible path $P$ is drawn as a straight-line through $T'$, $P$ crosses at most each of the $f-1$ other flexible paths.
Add one subdivision vertex to $P$ inside $T'$ for each crossing to ensure 1-planarity inside $T'$.

The whole drawing, including the static edges, is still $k$-planar:
A static edge crosses with a fixed flexible path in the original drawing if and only if it does in the simplified drawing.
Flexible edges inside a ``shrunken triangle'' cross other edges at most once,
and (potentially shortened) flexible edges outside a ``shrunken triangle'' received no additional crossings by the redrawing, hence remain crossed at most $k$ times.

Next, we argue that after redrawing, each flexible path $P$ has at most $(s^2 + 3s + 1) \cdot (2 + f - 1)$ internal vertices.
Applying \Cref{lemma:triangulations}, inside each of the at most $t(s) \leq s^2 + 3s + 1$ triangles $T$ of the triangulation,
a flexible path has, by construction, at most two internal vertices on the border of $T'$ plus at most $f-1$ internal vertices inside $T'$ if it intersects $T'$, and no internal vertices otherwise.  

Hence, in total, in the current drawing, a flexible path has at most $(s^2 + 3s + 1) \cdot (2 + f - 1)$ internal vertices, and thus length at most $(s^2 + 3s + 1) \cdot (2 + f - 1) + 1$.

Finally, to obtain a proper drawing of $G'$, we insert the required number of subdivision vertices for each flexible path that is too short arbitrarily.
\qed\end{proof}

Our result now follows analogously to the previous kernels.

\begin{theoremreused}{theorem:geom_k_planar_fen_fpt}
    \textsc{Geometric $k$-Planarity}, parameterized by the feedback edge number~$\ell$,
admits a kernel with $\mathcal{O}\big(2^{\mathcal{O}(3^{\ell}\log \ell)}\big)$ edges.
\end{theoremreused}

\begin{proof}
The correctness of this kernelization follows exactly as in the proof of \Cref{lemma:fen_correctness}, except that we now use the notion of ``very, very long paths'' and apply \Cref{lemma:triangle_redrawing} instead of \Cref{lemma:global_crossing_elimination} to obtain a drawing of the kernel from a drawing of the original graph.

Clearly, the kernel can be computed in polynomial time, and moreover,  \textsc{Geometric $k$-Planarity} is in $\exists\mathbb{R}$ and thus decidable \cite{geometric_k_planar_complexity}\footnote{The authors of \cite{geometric_k_planar_complexity} do not explicitly show membership, but rather implicitly state the claim. However, it is straightforward to show membership by providing a suitable $\exists\mathbb{R}$-sentence as was done for related problems such as \textsc{Geometric Thickness} in~\cite{geometric_thickness}.  }.
It remains to bound the kernel size.
Define $S_i \coloneqq \sum_{j=1}^{i} |E(P_j)|$, so that the kernel has $S_p$ edges.  
Note that $S_1 = \mathcal{O}(\ell)$, $p = \mathcal{O}(\ell)$, and  
$
S_i = \mathcal{O}( S_{i-1}^2 \cdot p ) = \mathcal{O}( S_{i-1}^3 )
$
for $1 < i \leq p$.  
Hence, the kernel has at most $\mathcal{O}(\ell^{ \mathcal{O}(3^p) }) = \mathcal{O}(2^{ \mathcal{O}(3^\ell \log \ell) })$ edges.
\end{proof}

\fi
\section{Lower Bounds}
\iflong
In this section we prove \Cref{theorem:npc_by_fvs,theorem:npc_by_bandwidth}.
\else
In this section we prove \Cref{theorem:npc_by_fvs}. The proof of \Cref{theorem:npc_by_bandwidth}, which establishes the \NP-completeness of \textsc{Geometric 1-Planarity} for bounded bandwidth, is deferred to the appendix. There, we also discuss why the known bandwidth-hardness for the topological case lifts to the geometric setting, whereas the hardness for feedback vertex set and pathwidth does not readily lift.
\fi
\iflong
For bandwidth, we lift the known \NP-hardness from the topological case~\cite{1planar_parameterized} to the geometric setting.
It is standard that bandwidth upper-bounds pathwidth, so hardness under bounded bandwidth already implies hardness under bounded pathwidth.
However, the existing bounded-bandwidth hardness for \textsc{1-Planarity} does not specify a concrete constant.
In contrast, our construction certifies a \emph{concrete} and \emph{small} bound of pathwidth at most~15.

For feedback vertex number, the known hardness in the topological case~\cite{k_planarity_gd} does not transfer by the same route \iflong(the analogue of \Cref{lemma:bandwidth_blowup} for pathwidth does not hold)\else(pathwidth does not stay bounded under replacement of edges with gadgets of constant size)\fi, and neither is it clear that Thomassen's characterization can be applied.
Thus we give a novel reduction from \textsc{Bin Packing}.
\fi

We use the following notion for both results.
A \emph{two-terminal edge gadget} is a graph $H$ with two distinguished \emph{attachment vertices} $\alpha,\beta$.
Given an edge $uv$ of a graph, \emph{replacing $uv$ by $H$} means taking a fresh copy of $H$, call it $H_{uv}$, and identifying its two attachment vertices with $u$ and $v$ (in either order). We call vertices $V(H_{uv})\setminus\{u,v\}$ and edges $E(H_{uv}) \setminus \{uv\}$ \emph{gadget-internal}.

\subsection{Feedback Vertex Number and Pathwidth}

The \emph{feedback vertex number} of a graph is the least number of vertices whose deletion makes the graph acyclic. For a definition of pathwidth, we refer to~\cite{fptbook}.

\paragraph{Bin Packing.}
We reduce from the strongly NP-hard \textsc{Bin Packing} problem \cite{GareyJohnson}, which asks whether, given a finite set $U$ of items with sizes $s(u) \in \mathbb{Z}^+$ for each $u \in U$, a bin capacity $B \in \mathbb{Z}^+$, and an integer $K > 0$, the set $U$ can be partitioned into disjoint subsets $U_1, U_2, \ldots, U_K$ such that $\sum_{u \in U_i} s(u) \le B$ for all $i \in [K]$.

We assume without loss of generality that $K \ge 2$, since the case $K = 1$ is trivial. Moreover, we may assume that $\sum_{u \in U} s(u) = K \cdot B$, meaning that all bins must be exactly filled. This can be ensured without changing the instance’s feasibility: if $\sum_{u \in U} s(u) < K \cdot B$, we can add $K \cdot B - \sum_{u \in U} s(u)$ dummy items of size $1$ each to obtain an equivalent instance, whereas if $\sum_{u \in U} s(u) > K \cdot B$, the instance is trivially infeasible.
Furthermore, we may assume that the minimum item size satisfies $\min_{u \in U} s(u) \geq K + 1$, for if not, we can multiply all item sizes and the bin capacity $B$ by $(K+1)$ to ensure this while preserving equivalence of the instance.

\paragraph{Reduction.}
Fix an instance of \textsc{Bin Packing} with the properties described above.
We construct an instance of \textsc{Geometric 1-Planarity} as follows.
We begin with the triconnected graph shown in \Cref{figure:hardness_example}, with distinguished vertices $s, t, r^\ell_1, r^\ell_2, r^r_1, r^r_2$, which we refer to as the \emph{frame}.
From a copy of $K_6$, we form a two-terminal edge gadget by selecting two distinct vertices, and we replace each edge of the frame with such a gadget.

Next, we add a path of length $|U| + K - 1$ between $r^\ell_1$ and $r^\ell_2$, referred to as the \emph{red left path}, and a path of length $K \cdot B$ between $r^r_1$ and $r^r_2$, referred to as the \emph{red right path}.
We then insert $K-1$ \emph{purple} edges from $s$ to vertices on the red right path so that the right path is subdivided into $K$ subpaths of length $B$ each.  
For each item $u \in U$, we create a fresh copy of $K_{2, s(u)}$, which we call the \emph{diamond} corresponding to $u$. We add an edge between $s$ and one vertex of the size-2 part of the diamond’s bipartition, call it the \emph{diamond vertex} $d_u$, and identify the other with $t$.
Call the union of $u$'s diamond and edge $sd_u$ the \emph{item gadget} of $u$. 
See \Cref{figure:hardness_example} for an illustration.

\looseness=-1
Intuitively, the triconnected frame reinforced with attached $K_6$-gadgets provides a rigid barrier that can essentially only be drawn as shown in the figure, and cannot be crossed by \Cref{lemma:uncrossable_k6}.
Hence, the remaining gadgets are forced to be drawn inside the frame.
The red right path, subdivided by the purple edges, models the bins (each is a subpath of length $B$). The red left path, together with the purple edges that must cross it, encodes the choice of bin for each item gadget.

\begin{figure}[t]
    \centering
    \includegraphics[page=1]{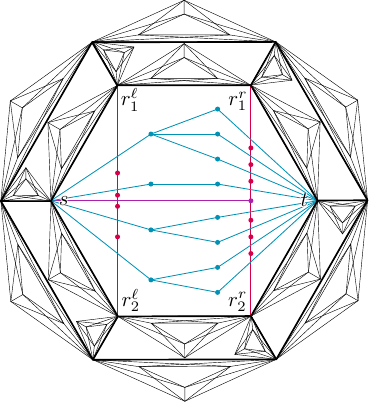}
    \caption{Example of our reduction for a \textsc{Bin Packing} instance $(U, K, B)$ with item sizes $3, 1, 2, 2$, number of bins $K = 2$, and bin capacity $B = 4$. The shown drawing of $G$ induces a solution with one bin containing items of sizes $3,1$, and the other bin with items of sizes $2,2$. 
    For illustrative purposes, we do not have $\min_{u \in U} s(u) \geq K + 1$\iflong, and the instance is not drawn in the exact style used in the backwards direction of the correctness proof\fi.
    The frame graph with distinguished vertices $s, t, r_1^\ell, r_2^\ell, r_1^r, r_2^r$ is drawn in bold black, the $K_6$ gadgets in black, left and right red paths are drawn in red, purple edges in purple, and item gadgets in blue.
    }
    \label{figure:hardness_example}
\end{figure}

\begin{lemma}[\textnormal{Adapted from \protect\cite[Lemma~5.1]{original_np_hardness_1_planar}}]\label{lemma:uncrossable_k6}
In any 1-planar embedding of $K_6$, between any two vertices there exists a path such that all edges in that path are crossed.
\end{lemma}

\begin{lemma}\label{lemma:hardness_correctness}
Graph $G$ is geometric 1-planar if and only if the given bin packing instance is positive.
\end{lemma}

\begin{proof}

$(\Rightarrow):$
Assume there is a geometric 1-planar embedding $\varepsilon$ of $G$.

\smallskip
\emph{Properties ensured by the frame-construction.}
For each edge $uv$ of the frame,
we can by \Cref{lemma:uncrossable_k6}
find a path from $u$ to $v$ in the $K_6$ attached to $uv$, such that all edges of the path cross with edges of this copy of $K_6$.

The subembedding $\varepsilon'$ of $\varepsilon$ induced by taking one such path for each edge of the frame is thus planar.
Observe that the graph underlying $\varepsilon'$
is, by construction, a subdivison of a triconnected graph (the frame graph depicted in \Cref{figure:hardness_example}).
Therefore, as is well-known, the faces of $\varepsilon'$ are uniquely determined (up to mirror image).
In particular, vertices $s, t, r^\ell_1, r^\ell_2, r^r_1, r^r_2$ lie on a shared face in $\varepsilon'$, call it $R \subseteq \mathbb{R}^2$.
Further, $R$ is the only face in $\varepsilon'$ that contains both $s$ and $t$ (resp. $r^\ell_1$ and $r^\ell_2$, resp. $r^r_1$ and $r^r_2$).

In total, this gives us that all edges of $G$ that are not frame-edges or $K_6$ gadgets are entirely contained inside region $R$ in the original embedding $\varepsilon$.

\smallskip
\emph{Forced crossings.}
Fix an item $u \in U$ and consider its diamond.
At least one red path \emph{crosses the diamond fully}, by which we mean it crosses $\ge s(u)$ edges of the diamond:
If one red path crosses edge $sd_u$, the other red path cannot cross $sd_u$ as well and thus needs to fully cross the diamond.
If no red path crosses edge $sd_u$, both red paths need to cross the diamond fully.

We show that the right red path crosses all diamonds fully.
Suppose for one item $u \in U$, the left red path crosses $u$'s diamond fully.
Since $s(u) \geq K + 1$, this produces at least $K + 1$ crossings.
Each item gadget for items $u' \neq u$ clearly produces at least one crossing with the left red path as well.
Hence the left red path has at least $K + 1 + |U| - 1 = |U| + K$ crossings with item gadgets, which is impossible, since it consists of only $|U| + K - 1$ edges.
Hence for each $u \in U$, the left red path does not cross $u$'s diamond fully.
By the above, this means the right red path crosses all diamonds fully, i.e., it is crossed by at least $\sum_{u\in U} s(u)$ item-gadget edges.
Since it consists of $\sum_{u\in U} s(u)$ edges,
this means each of its edges is crossed by an item-gadget edge.

Thus the red paths cannot cross each other.
Hence each purple edge crosses the left path.
This also means the left red path is saturated with crossings, since there are $K - 1$ purple edges, each item gadget crosses the left red path at least once, and the left red path counts $|U| + K - 1$ edges.

\smallskip
\emph{Defining a bin packing.}
Divide the region $R$ into open regions $R^\ell, R^m, R^r$,
such that $R^\ell$, the \emph{left region}, has $s$ on the boundary, $R^m$, the \emph{middle region}, is bounded by the red paths, 
and $R^r$, the \emph{right region}, has $t$ on the boundary. This is well-defined as the red paths do not cross each other.

Each diamond vertex $d_u$ lies in $R^m$:
If diamond vertex $d_u$ were in $R^r$, edge $sd_u$ would cross both red paths, which is impossible.
If $d_u$ were in $R^\ell$, the left red path would cross $u$'s diamond fully, which we have seen to be impossible.

This also implies that not only is the left red path saturated with crossings from item-gadgets and purple edges, but, more specifically, each crossing is either with a purple edge or an edge of the form $sd_u$ for some $u \in U$.

\looseness=-1
Divide $R^m$ into open regions $R^m_1, \dots, R^m_K$ by subtracting the $K-1$ purple edge-segments from $R^m$. Since all purple edges cross the left red path, this is well-defined.
Hence, each diamond vertex $d_u$ for $u \in U$ lies in precisely one $R^m_i$.
Let $U_1, \dots, U_K$ be the partition of $U$ induced by assigning $u \in U$ to $U_i$ if $d_u$ is in $R^m_i$.

\smallskip
\emph{Showing the packing is a solution.}
We claim $U_1, \dots, U_K$ is a solution to the \textsc{Bin Packing} instance, i.e., we have $\sum_{u \in U_i} s(u) \le B$ for all $i \in [K]$.

Let $i \in [K]$.
We aim to show $\sum_{u \in U_i} s(u) \leq B$.
Let $u \in U_i$.
By construction, $d_u$ lies in $R^m_i$.
Region $R^m_i$ is bounded by two purple edges, the left red path, and the right red path.
The bounding purple edges cross the left red path.
The left red path crosses only with edges of the form $sd_{u'}$ for $u' \in U$.
Thus, the only way to route the edges of $u$'s diamond to $t$ is through the red right path, of which precisely $B$ edges lie on the boundary of $R^m_i$.
Since we know each diamond crosses the right red path fully, $u$'s diamond crosses this portion of the right red path at least $s(u)$ times.

Hence in total, the items in $U_i$ produce at least $\sum_{u \in U_i} s(u)$ crossings in the 
length-$B$ subpath of the right red path bounding $R^m_i$.
But since each of these $B$ edges can only be crossed at most once, we obtain, as required, 
$\sum_{u \in U_i} s(u) \leq B$.

$(\Leftarrow):$
Let $U_1, \dots, U_K$ be a solution to the bin packing instance.
We create a geometric 1-planar embedding $\varepsilon$ of $G$ as follows.
First, draw the frame and $K_6$ gadgets like is displayed in \Cref{figure:hardness_example}.
We draw both red paths as straight lines, but do not fix the position of the subdivision vertices for now, except for the $K-1$ subdivision vertices of the right red path that are adjacent to purple edges, which we position by dividing the line segment from $r^r_1$ to $r^r_2$ into $K$ equidistant segments.

Next, we insert the item gadgets.
For each $i \in [K]$, consider the $i$'th equidistant segment from above, and shift it an arbitrarily small distance to the left.

For each $u \in U_i$, place vertex $d_u$ at a distinct (inner) point of the shifted segment.
To draw the diamond of $u$, consider the line segment from $d_u$ to $t$. Observe that one can draw the diamond within an arbitrarily narrow tunnel around this segment, and can thus avoid all crossings except for with the right red path.

Finally, we can fix the positions of the subdivision vertices of the red paths.

The left red path is crossed by all $K-1$ purple edges and the edge $sd_u$ for each $u \in U$.
As the left red path has length $|U| + K - 1$, we can insert the paths subdivision vertices on the line such that each edge of the path is crossed exactly once.

For the right red path, we need to divide each equidistant segment into $B$ segments, ensuring the right red path has length $K \cdot B$ in total.
By construction, each such equidistant segment is crossed by $\sum_{u \in U_i} s(u) = B$ item gadget edges.
Thus, we can insert the subdivision vertices at appropriate positions such that each edge of the right red path is crossed exactly once.

In total, this gives the required geometric 1-planar embedding of $G$.
\qed\end{proof}

Since deleting the 12 frame vertices from $G$ yields a disjoint union of only $K_4$'s, paths, and stars, we have: 
\begin{lemma}\label{lemma:bounded_pathwidth_and_fvn}
$G$ has pathwidth $\leq 15$ and feedback vertex number $\leq 48$.
\end{lemma}
\begin{proof}
    Deleting the 12 frame vertices from $G$ yields a disjoint union of stars (stemming from the item-gadgets), paths (stemming from the left and right red paths), and 18 $K_4$'s (stemming from the $K_6$ gadgets).
    This disjoint union has pathwidth 3, hence $G$ has pathwidth at most $12 + 3 = 15$.
    Moreover, after deleting the 12 frame vertices, deleting 2 vertices per $K_4$ yields an acyclic graph.
    Hence, the feedback vertex number of $G$ is at most $12 + 18 \cdot 2 = 48$.
\qed\end{proof}

The reduction can be computed in polynomial time, since \textsc{Bin Packing} is strongly \NP-hard~\cite{GareyJohnson}, that is, we can assume the numeric inputs to be encoded in unary.
Finally, recall that \textsc{Geometric 1-Planarity} is in \NP~\cite{rediscovered}.
Thus, by \Cref{lemma:hardness_correctness,lemma:bounded_pathwidth_and_fvn}, we obtain

\begin{theoremreused}{theorem:npc_by_fvs}

    \textsc{Geometric 1-Planarity} remains \NP-complete
    for instances of pathwidth at most~15 or feedback vertex number at most~48.
    
\end{theoremreused}

\iflong
\subsection{Bandwidth}

Let $G$ be a graph with $n$ vertices. The \emph{bandwidth} of $G$, denoted $\mathrm{bw}(G)$, is
\[
 \min\Big\{\max_{uv\in E(G)}\left|\sigma(u)-\sigma(v)\right|\ \mid\ \sigma:V(G)\to[n]\ \text{is a bijection}\Big\},
\]
where we call $\left|\sigma(u)-\sigma(v)\right|$ the \emph{span} of $uv$ in $\sigma$.
In other words, a graph has bandwidth $\leq b$ when we can arrange its vertices on a line with integer coordinates so that adjacent vertices have distance at most $b$.

The key structural insight underlying our hardness result is that bounded bandwidth remains stable under edge replacements by constant-size two-terminal edge gadgets.
\begin{lemma}\label{lemma:bandwidth_blowup}
Let $G$ be a graph with $\mathrm{bw}(G)=b$, and let $H$ be a two-terminal edge gadget on $t$ vertices.
Let $G_H$ be obtained by replacing every edge of $G$ by $H$.
Then
\[
\mathrm{bw}(G_H)\ \le\ (b+1)\,\bigl(1+(t-2)b\bigr).
\]
\end{lemma}

\begin{proof}
Fix a bandwidth-$b$ ordering $\sigma$ of $G$.
For each $x\in V(G)$, form a \emph{column}
\[
C_x\ \coloneqq\ \{x\}\ \cup\!\!\!\!\bigcup_{\substack{xy\in E(G)\\ \sigma(x)<\sigma(y)}}\!\!\!\!
\bigl(V(H_{xy})\setminus\{y\}\bigr).
\]
List the columns in the increasing order of $\sigma(x)$, and order vertices arbitrarily within each column to obtain an ordering $\sigma^\star$ of $V(G_H)$.

As $\sigma$ is a bandwidth-$b$ ordering of $G$, each $x\in V(G)$ has at most $b$ neighbors~$y$ with $\sigma(y)>\sigma(x)$.
Each such edge contributes at most $t-2$ gadget-internal vertices to~$C_x$.
Hence, for all $x \in V(G)$,
\[
|C_x|\ \le\ c\ \coloneqq\ 1+(t-2)b.
\]

Any edge of $G_H$ with both endpoints in one column (all gadget-internal edges not incident to the higher-index endpoint) has span at most $c$ in $\sigma^\star$.
The only inter-column edges of $G_H$ in $\sigma^\star$ are those in some $H_{uv}$ incident to the higher-index endpoint $v$.
Let $u'v$ be such an edge,
and let $C_x$ be the column containing~$u'$.
Since $|\sigma(v)-\sigma(x)|\le b$, the endpoints $u', v$ lie within a block of at most $b+1$ consecutive columns from $C_x$ through $C_v$, each consisting of at most~$c$ vertices.
Hence the span of $u'v$ in $\sigma^\star$ is at most $(b+1)c$.
In total, each edge of $G_H$ has span at most $\max \{c,(b+1)c \} = (b+1)c$ in $\sigma^\star$, and we obtain
\[
\mathrm{bw}(G_H)\ \le\ (b+1)c \ =\ (b+1)\bigl(1+(t-2)b\bigr). \tag*{\qed}
\]
\end{proof}

\Cref{lemma:bandwidth_blowup} allows us to lift the known \NP-hardness of \textsc{1-Planarity} for bounded-bandwidth graphs to the geometric setting, using the polynomial-time reduction from \textsc{1-Planarity} to \textsc{Geometric 1-Planarity} by Schaefer~\cite{geometric_1_planarity_np_hard}.
\npcbybandwidth*
\begin{proof}
    There is a polynomial-time reduction from \textsc{1-Planarity} to \textsc{Geometric 1-Planarity}, obtained by replacing each edge of the input graph with the fixed two-terminal edge gadget shown in \Cref{figure:bandwidth_gadget} \cite[Theorem~2]{geometric_1_planarity_np_hard}.

    \begin{figure}[t]
    \centering
    \includegraphics[page=5]{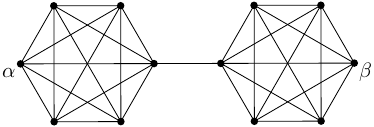}
    \caption{The two-terminal edge gadget with attachment vertices $\alpha, \beta$ used in the \NP-hardness proof of \textsc{Geometric 1-Planarity} by Schaefer~\cite{geometric_1_planarity_np_hard}.}
    \label{figure:bandwidth_gadget}
\end{figure}

    Moreover, it is known that \textsc{1-Planarity} remains \NP-complete even when restricted to graphs of bandwidth at most some fixed constant~$c_1$ \cite[Theorem~4]{1planar_parameterized}. 
    Applying the above reduction to this restricted fragment, and invoking \Cref{lemma:bandwidth_blowup}, we obtain instances of \textsc{Geometric 1-Planarity} whose bandwidth is bounded by another constant~$c_2$.
    Thus, \NP-hardness is established.
    Finally, recall that \textsc{Geometric 1-Planarity} is in \NP~\cite{rediscovered}.
\qed\end{proof}

\fi
\section{Conclusion}

We gave a comprehensive set of results: a fixed-parameter algorithm for \textsc{Geometric 1-Planarity} parameterized by treedepth; subfactorial kernels parameterized by the feedback edge number for \textsc{1-Planarity}, \textsc{$k$-Planarity}, and \textsc{Geometric 1-Planarity}; a kernel for \textsc{Geometric $k$-Planarity} under the same parameterization; and matching \NP-completeness for bounded pathwidth, feedback vertex number, and bandwidth.

We improved the kernel for \textsc{$k$-Planarity}
parameterized by the feedback edge number $\ell$ from $\mathcal{O}((3\ell)!)$ to $\mathcal{O}(\ell\cdot 8^{\ell})$. Is a polynomial kernel possible?

A natural next step is to consider the (parameterized) complexity of \textsc{Geometric $k$-Planarity}. In the topological case, the leap from $1$ to $k$ is trivial for treedepth+$k$ and  feedback edge number: topological $k$-planarity reduces to 1-planarity by replacing each edge with a path of length $k$, which does not blow up these parameters~\cite{k_planarity_gd}. In the geometric case, however, this is far from trivial, as for $k\ge 2$, no analogue of Thomassen’s characterization is known, and such a characterization is in a certain natural sense even impossible for $k \geq 3$~\cite{nagamochi2013straight}. For feedback edge number, we sidestepped this issue and gave an argument that does not rely on such a characterization.

For treedepth, this seems challenging. As a first step, for triconnected graphs, \textsc{Geometric $k$-Planarity} parameterized by treedepth+$k$ is easily seen to be \FPT: only processing akin to \Cref{rule:I} is necessary, and the underlying counting argument can be lifted easily. Does this tractability extend to general graphs in the \emph{uniform} sense? (Note that existence of a \emph{non-uniform} \FPT{} algorithm is always guaranteed for hereditary properties~\cite{sparsity}; in particular, (geometric) $k$-planarity is hereditary for each fixed $k$.)

More fundamentally, for $k\ge 2$, the complexity of \textsc{Geometric $k$-Planarity} is not well understood.
While \textsc{Geometric 1-Planarity} is in \NP~\cite{rediscovered}, it is known that at some point (unless $\NP=\exists\mathbb{R}$), the problem ceases to even be in \NP, as \textsc{Geometric 867-Planarity} is $\exists\mathbb{R}$-complete~\cite{geometric_k_planar_complexity}. When does this shift occur? Is \textsc{Geometric 2-Planarity} in \NP?

\bibliographystyle{splncs04}
\bibliography{refs}

\begin{thebibliography}{10}
\providecommand{\url}[1]{\texttt{#1}}
\providecommand{\urlprefix}{URL }
\providecommand{\doi}[1]{https://doi.org/#1}

\bibitem{rectilinear_crossing_number_bipartite}
{\'{A}}brego, B.M., Dondzila, K., Fern{\'{a}}ndez{-}Merchant, S., Lagoda, E.,
  Sajjadi, S., Sapozhnikov, Y.: On the rectilinear local crossing number of
  \emph{K\({}_{\mbox{m, n}}\)}. J. Inf. Process.  \textbf{25},  542--550
  (2017). \doi{10.2197/IPSJJIP.25.542}

\bibitem{rectilinear_crossing_number_complete}
{\'{A}}brego, B.M., Fern{\'{a}}ndez{-}Merchant, S.: The rectilinear local
  crossing number of k\({}_{\mbox{n}}\). J. Comb. Theory {A}  \textbf{151},
  131--145 (2017). \doi{10.1016/J.JCTA.2017.04.003}

\bibitem{grid_drawings}
Alam, M.J., Brandenburg, F.J., Kobourov, S.G.: Straight-line grid drawings of
  3-connected 1-planar graphs. In: Wismath, S.K., Wolff, A. (eds.) Graph
  Drawing - 21st International Symposium, {GD} 2013, Bordeaux, France,
  September 23-25, 2013, Revised Selected Papers. Lecture Notes in Computer
  Science, vol.~8242, pp. 83--94. Springer (2013).
  \doi{10.1007/978-3-319-03841-4\_8}

\bibitem{1planar_parameterized}
Bannister, M.J., Cabello, S., Eppstein, D.: Parameterized complexity of
  1-planarity. J. Graph Algorithms Appl.  \textbf{22}(1),  23--49 (2018).
  \doi{10.7155/JGAA.00457}

\bibitem{computational_geometry_book}
de~Berg, M., Cheong, O., van Kreveld, M.J., Overmars, M.H.: Computational
  geometry: algorithms and applications, 3rd Edition. Springer (2008).
  \doi{10.1007/978-3-540-77974-2}

\bibitem{ic_planar}
Brandenburg, F.J., Didimo, W., Evans, W.S., Kindermann, P., Liotta, G.,
  Montecchiani, F.: Recognizing and drawing ic-planar graphs. Theor. Comput.
  Sci.  \textbf{636},  1--16 (2016). \doi{10.1016/J.TCS.2016.04.026}

\bibitem{crossing_types}
Cabello, S., Dobler, A., Fijavz, G., Hamm, T., Wagner, M.H.: A dichotomy for
  1-planarity with restricted crossing types parameterized by treewidth. In:
  Chen, H., Hon, W., Tsai, M. (eds.) 36th International Symposium on Algorithms
  and Computation, {ISAAC} 2025, Tainan, Taiwan, December 7-10, 2025. LIPIcs,
  vol.~359, pp. 16:1--16:15. Schloss Dagstuhl - Leibniz-Zentrum f{\"{u}}r
  Informatik (2025). \doi{10.4230/LIPICS.ISAAC.2025.16},
  \url{https://doi.org/10.4230/LIPIcs.ISAAC.2025.16}

\bibitem{constrained_delaunay}
Chew, L.P.: Constrained delaunay triangulations. Algorithmica  \textbf{4}(1),
  97--108 (1989). \doi{10.1007/BF01553881}

\bibitem{fptbook}
Cygan, M., Fomin, F.V., Kowalik, L., Lokshtanov, D., Marx, D., Pilipczuk, M.,
  Pilipczuk, M., Saurabh, S.: Parameterized Algorithms. Springer (2015).
  \doi{10.1007/978-3-319-21275-3}

\bibitem{geometric_thickness}
Depian, T., Fink, S.D., Firbas, A., Ganian, R., N{\"{o}}llenburg, M.: Pathways
  to tractability for geometric thickness. In: Kr{\'{a}}lovic, R.,
  Kurkov{\'{a}}, V. (eds.) {SOFSEM} 2025: Theory and Practice of Computer
  Science - 50th International Conference on Current Trends in Theory and
  Practice of Computer Science, Bratislava, Slovak Republic, January 20-23,
  2025, Proceedings, Part {I}. Lecture Notes in Computer Science, vol. 15538,
  pp. 209--224. Springer (2025). \doi{10.1007/978-3-031-82670-2\_16}

\bibitem{density_geometric_1_planar}
Didimo, W.: Density of straight-line 1-planar graph drawings. Inf. Process.
  Lett.  \textbf{113}(7),  236--240 (2013). \doi{10.1016/J.IPL.2013.01.013}

\bibitem{diestel}
Diestel, R.: Graph Theory, 4th Edition, Graduate texts in mathematics,
  vol.~173. Springer (2012). \doi{10.1007/978-3-662-70107-2}

\bibitem{Eggleton1986}
Eggleton, R.B.: Rectilinear drawings of graphs. Utilitas Mathematica
  \textbf{29},  149--172 (1986)

\bibitem{extending1planar}
Eiben, E., Ganian, R., Hamm, T., Klute, F., N{\"{o}}llenburg, M.: Extending
  partial 1-planar drawings. In: Czumaj, A., Dawar, A., Merelli, E. (eds.) 47th
  International Colloquium on Automata, Languages, and Programming, {ICALP}
  2020, July 8-11, 2020, Saarbr{\"{u}}cken, Germany (Virtual Conference).
  LIPIcs, vol.~168, pp. 43:1--43:19. Schloss Dagstuhl - Leibniz-Zentrum
  f{\"{u}}r Informatik (2020). \doi{10.4230/LIPICS.ICALP.2020.43}

\bibitem{GareyJohnson}
Garey, M.R., Johnson, D.S.: Computers and Intractability: {A} Guide to the
  Theory of NP-Completeness. W. H. Freeman (1979)

\bibitem{k_planarity_gd}
Gima, T., Kobayashi, Y., Okada, Y.: Structural parameterizations of
  $k$-planarity. In: Graph Drawing and Network Visualization (GD 2025) (2025),
  to appear

\bibitem{original_np_hardness_1_planar}
Grigoriev, A., Bodlaender, H.L.: Algorithms for graphs embeddable with few
  crossings per edge. Algorithmica  \textbf{49}(1),  1--11 (2007).
  \doi{10.1007/S00453-007-0010-X}

\bibitem{PartiallyPredrawnCrossingNumber}
Hamm, T., Hlinen{\'{y}}, P.: Parameterised partially-predrawn crossing number.
  In: Goaoc, X., Kerber, M. (eds.) 38th International Symposium on
  Computational Geometry, SoCG 2022, June 7-10, 2022, Berlin, Germany. LIPIcs,
  vol.~224, pp. 46:1--46:15. Schloss Dagstuhl - Leibniz-Zentrum f{\"{u}}r
  Informatik (2022). \doi{10.4230/LIPICS.SOCG.2022.46}

\bibitem{rediscovered}
Hong, S., Eades, P., Liotta, G., Poon, S.: F{\'{a}}ry's theorem for 1-planar
  graphs. In: Gudmundsson, J., Mestre, J., Viglas, T. (eds.) Computing and
  Combinatorics - 18th Annual International Conference, {COCOON} 2012, Sydney,
  Australia, August 20-22, 2012. Proceedings. Lecture Notes in Computer
  Science, vol.~7434, pp. 335--346. Springer (2012).
  \doi{10.1007/978-3-642-32241-9\_29}

\bibitem{reembedding_crossing_pairs}
Hong, S., Nagamochi, H.: Re-embedding a 1-plane graph into a straight-line
  drawing in linear time. In: Hu, Y., N{\"{o}}llenburg, M. (eds.) Graph Drawing
  and Network Visualization - 24th International Symposium, {GD} 2016, Athens,
  Greece, September 19-21, 2016, Revised Selected Papers. Lecture Notes in
  Computer Science, vol.~9801, pp. 321--334. Springer (2016).
  \doi{10.1007/978-3-319-50106-2\_25}

\bibitem{survey}
Kobourov, S.G., Liotta, G., Montecchiani, F.: An annotated bibliography on
  1-planarity. Comput. Sci. Rev.  \textbf{25},  49--67 (2017).
  \doi{10.1016/J.COSREV.2017.06.002}

\bibitem{nagamochi2013straight}
Nagamochi, H.: Straight-line drawability of embedded graphs. Tech. rep.,
  Technical Reports 2013-005, Department of Applied Mathematics and Physics~…
  (2013)

\bibitem{sparsity}
Ne{\v{s}}et{\v{r}}il, J., de~Mendez, P.O.: Sparsity - Graphs, Structures, and
  Algorithms, Algorithms and combinatorics, vol.~28. Springer (2012).
  \doi{10.1007/978-3-642-27875-4}

\bibitem{reidemeister}
Reidemeister, K.: Elementare {B}egr{\"u}ndung der {K}notentheorie. In:
  Abhandlungen aus dem Mathematischen Seminar der Universit{\"a}t Hamburg.
  vol.~5, pp. 24--32. Springer (1927)

\bibitem{schaefer2012graph}
Schaefer, M.: The graph crossing number and its variants: A survey. The
  electronic journal of combinatorics  (2012). \doi{10.37236/2713}

\bibitem{geometric_1_planarity_np_hard}
Schaefer, M.: Picking planar edges; or, drawing a graph with a planar subgraph.
  In: Duncan, C.A., Symvonis, A. (eds.) Graph Drawing - 22nd International
  Symposium, {GD} 2014, W{\"{u}}rzburg, Germany, September 24-26, 2014, Revised
  Selected Papers. Lecture Notes in Computer Science, vol.~8871, pp. 13--24.
  Springer (2014). \doi{10.1007/978-3-662-45803-7\_2}

\bibitem{geometric_k_planar_complexity}
Schaefer, M.: Complexity of geometric k-planarity for fixed k. J. Graph
  Algorithms Appl.  \textbf{25}(1),  29--41 (2021). \doi{10.7155/JGAA.00548}

\bibitem{straightening_characterization}
Thomassen, C.: Rectilinear drawings of graphs. J. Graph Theory  \textbf{12}(3),
   335--341 (1988). \doi{10.1002/JGT.3190120306}

\end{thebibliography}

\end{document}